\documentclass{article}

\usepackage{bussproofs}
\usepackage{cmll}
\usepackage{xspace}
\usepackage{amssymb}
\usepackage{amsthm}
\usepackage{amsmath}
\usepackage{color}

\newtheorem{theorem}{Theorem}
\newtheorem{lemma}[theorem]{Lemma}
\newtheorem{corollary}[theorem]{Corollary}
\newtheorem*{remark}{Remark}

\newtheorem{definition}{Definition}
\newtheorem{prop}[theorem]{Proposition}

\newcommand{\A}{A}
\newcommand{\B}{B}

\newcommand{\biat}{{\sf E}}

\renewcommand{\phi}{\varphi}
\newcommand{\tensor}{\otimes}
\renewcommand{\implies}{\multimap}
\renewcommand\t{\vdash}

\renewcommand{\H}{\ensuremath{\mathsf{H}}\xspace}

\newcommand{\HMILL}{\ensuremath{\mathsf{H\mbox{-}MILL}}\xspace}
\newcommand{\ILL}{\ensuremath{\mathsf{ILL}}\xspace}
\newcommand{\MILL}{\ensuremath{\mathsf{MILL}}\xspace}
\newcommand{\PCMILL}{\ensuremath{\mathsf{PCMILL}}\xspace}
\newcommand{\RSBIAT}{\ensuremath{\mathsf{RSBIAT}}\xspace}
\newcommand{\SRSBIAT}{\ensuremath{\mathsf{SRSBIAT}}\xspace}
\newcommand{\HRSBIAT}{\ensuremath{\mathsf{H\mbox{-}RSBIAT}}\xspace}

\newcommand\AC[1]{\AxiomC{$#1$}}
\newcommand\rl[1]{\RightLabel{#1}}
\newcommand\BC[1]{\BinaryInfC{$#1$}}
\newcommand\UC[1]{\UnaryInfC{$#1$}}
\newcommand\dip{\DisplayProof}

\newcommand\dnote[1]{#1}

\begin{document}


\title{Non-normal modalities in variants of Linear Logic\thanks{This article greatly extends an earlier paper which appeared in the proceedings of the 21st European Conference on Artificial Intelligence, ECAI 2014, \cite{PorelloTroquardECAI2014}.}}

\author{Daniele Porello\footnote{Laboratory for Applied Ontology, ISTC-CNR, Trento, Italy. \texttt{daniele.porello@loa.istc.cnr.it}} \and Nicolas Troquard\footnote{Algorithmic, Complexity and Logic Laboratory, Universit\'e Paris-Est, Cr\'eteil, France. \texttt{troquard@loa.istc.cnr.it}}
}

\date{}
\maketitle

\begin{abstract}
This article presents modal versions of resource-conscious logics. We
concentrate on extensions of variants of Linear Logic with one minimal
non-normal modality. In earlier work, where we investigated agency in
multi-agent systems, we have shown that the results scale up to logics
with multiple non-minimal modalities. Here, we start with the language
of propositional intuitionistic Linear Logic without the additive
disjunction, to which we add a modality. We provide an interpretation
of this language on a class of Kripke resource models extended with a
neighbourhood function: modal Kripke resource models. We propose a
Hilbert-style axiomatization and a Gentzen-style sequent calculus. We
show that the proof theories are sound and complete with respect to
the class of modal Kripke resource models. We show that the sequent
calculus admits cut elimination and that proof-search is in PSPACE. We
then show how to extend the results when non-commutative connectives
are added to the language. Finally, we put the logical framework to
use by instantiating it as logics of agency. In particular, we propose
a logic to reason about the resource-sensitive use of artefacts and
illustrate it with a variety of examples.
\medskip

\end{abstract}

\section{Introduction}

Logics for resources and modalities each got their share of the
attention and have also already been studied together. Besides the seminal work on intuitionistic modal logic in~\cite{PlotkinStirlingTARK1986,SimpsonThesis1994}, modalities for substructural implications have been studied in \cite{DagostinoEtAlSL1997} and extensions of intuitionistic Linear Logic with modalities have been investigated for
example in \cite{marion04,KamideTCS2006}. Moreover, modal
versions of logics for resources that are related to Linear Logic have
been provided in~\cite{PymTCS2004,Pym2006,courtault}.

Modalities in substructural logics are thus not new, although with
one important detail: to our knowledge, modalities in sub-structural
logics have always been restricted to \emph{normal}
modalities. Modalities that are not normal have been confined to the
realm of classical logic.

\emph{Normal modalities} are the modalities within a logic that is at
least as strong as the standard modal logic K. \emph{Non-normal
  modalities} on the other hand are the modalities that fail to satisfy some of the principles of the standard modal logic K. They cannot be
evaluated over a Kripke semantics but one typical semantics,
neighbourhood models, rely on possible worlds. They were introduced
independently by Scott and Montague. Early results were offered by
Segerberg. Chellas built upon and gave a textbook presentation
in~\cite{chellas80}.

The significance of non-normal modal logics and their semantics in
modern developments in logics of agents has been emphasised
before~\cite{DBLP:journals/sLogica/Arlo-CostaP06}.  Indeed many logics
of agents are non-normal and neighbourhood semantics allows for
defining the modalities that are required to model a number of
application domains: logics of coalitional
power~\cite{pauly02coalitional}, epistemic logics without
omniscience~\cite{Vardi86,DBLP:journals/ai/LismontM94}, logics of
agency~\cite{Governatori05Elgesem}, etc.

\medskip
Let us briefly present some features of reasoning about agency
specifically. Modalities of agency aimed at modelling the result of an
action have been largely studied in the literature in practical
philosophy and in multi-agent systems \cite{kangerkanger66,
  Prn77actionsocial, Governatori05Elgesem, belnap01facing,Tr14jaamas}. Logics of agency assume a set of
agents, and to each agent $i$ associate a modality 
$\biat_i$. We read $\biat_i \phi$ as ``agent $i$ brings about
$\phi$''. With the notable exception of Chellas'
logic~\cite{chellas69imperatives}, it is generally accepted in logics
of agency that no agent ever brings about a tautology. This is because
when $i$ brings about $\phi$, it is intended that $\phi$ might have
not be the case if it were not for $i$'s very agency. Classically, it
corresponds to the axiom $\lnot \biat_i \top$.  This principle is
enough for the logic of the modality $\biat_i$ to be non-normal as it is inconsistent with the necessitation rule. A
normal modal logic for agency would also dictate that if $i$ does
$\phi$ then she also does $\phi \lor \psi$. This is refuted in general
on the same grounds as Ross' paradox.
For instance, purposefully doing that I am rich
should not imply purposefully doing that I am rich or unhealthy.

In classical logic, $\biat_i \phi$ allows one to capture that agent
$i$ brings about \emph{the state of affairs} $\phi$. Moving from
states of affairs to resources, it is then interesting to lift these
modalities from classical logic to resource-conscious logic. In doing
so, one can capture with $\biat_i A$ that agent $i$ brings about
\emph{the resource} $A$.

\medskip
To make a start with this research program, we will combine
intuitionistic fragments of Linear Logic with non-normal
modalities. Linear Logic~\cite{Girard1987} is a resource-conscious
logic that allows for modelling the constructive content of deductions
in logic.
An \emph{intuitionistic} version of Linear Logic, as the one we will
work with, has a number of desirable features. One that is simple and
yet greatly appreciated is that in intuitionistic
sequent calculus every sequent has a single ``output'' formula. This
feature favours the modelling of input-output processes. It will prove
particularly adequate for our application to agency and artefacts as
it provides a simple mechanism for the compositionality of individual
artefacts' functions into complex input-output processes.

\medskip

The resource-sensitive nature of Linear Logic can be viewed as the absence of structural rules in the sequent calculus. Linear Logic rejects the global validity of weakening (W), that amounts to a monotonicity of the entailment, and contraction (C), that is responsible for arbitrary duplications of formulas, e.g., $A \rightarrow A \wedge A$ is a tautology in classical logic, and so is $A \land A \rightarrow A$.

\begin{center}
\begin{tabular}{ccc}
\AC{\Gamma \t A}
\rl{(W)}
\UC{\Gamma, B \t A}
\dip

&

\AC{\Gamma, B, B \t A}
\rl{(C)}
\UC{\Gamma, B \t A}
\dip

&
\AC{\Gamma, A, B \t C}
\rl{(E)}
\UC{\Gamma, B, A \t C}
\dip

\end{tabular}
\end{center}

 Accordingly, the linear implication $\implies$ encodes
 resource-sensitive deductions, for example from $A$ and $A \implies
 B$ we can infer $B$, by \emph{modus ponens}, but we are not allowed
 to conclude $B$ from $A$, $A$, and $A\implies B$.  Exchange (E) still
 holds. (Although we will later restrict it.) Hence, contexts of
 formulas $\Gamma$ in sequent calculus are considered multisets. By
 dropping weakening and contraction, we are led to define two
 non-equivalent conjunctions with different behaviours: the
 multiplicative conjunction $\otimes$ (tensor) and the additive
 conjunction $\with$ (with). The intuitive meaning of $\otimes$ is
 that an action of type $A \otimes B$ can be performed by summing the
 resources that are relevant to perform $A$ and to perform $B$.  The
 unit $\textbf{1}$ is the neutral element for $\otimes$ and can
 represent a null action. A consequence of the lack of weakening
 is that $A \otimes B$ no longer implies $A$, namely the resources
 that are relevant to perform $A \otimes B$ may not be relevant to
 perform just $A$. The absence of contraction means that $A \implies A
 \otimes A$ is no longer valid. The additive conjunction $A\with B$
 expresses an option, the choice to perform $A$ or $B$. Accordingly $A
 \with B \implies A$ and  $A \with B \implies B$ hold in Linear Logic, the resources that enable
 the choice between $A$ and $B$ are relevant also to make $A$ or to
 make $B$. The linear implication $A \implies B$ expresses a form of
 causality, for example ``If I strike a match, I light the room''
 the action of striking a match is consumed, in the sense that it is
 no longer available after the room is lighted. Linear implication and
 multiplicative conjunction interact naturally so that $(A \otimes (A
 \implies B)) \implies B$ is valid.

Linear Logic operators have been applied to a number of topics in knowledge representation and multiagent systems such as planning \cite{KanovichVauzeilles2001}, preference representation and resource allocation \cite{HarlandWinikoffAAMAS2002,PorelloEndrissKR2010,PorelloEndrissECAI2010}, social choice \cite{PorelloIJCAI2013}, actions modelling \cite{BorgoEtAlFOIS2014}, and narrative generation \cite{MartensEtAL2013}.

These propositional operators are very useful when talking about
resources and agency. They allow one to capture the following notions:
\begin{itemize}
\item Bringing about both $A$ and $B$ together: $\biat_i (A \otimes B)$. In such a way that $\biat_i (A \otimes B)$ implies $A\otimes B$, but does not imply $A$ alone.
\item Bringing about an option between $A$ and $B$: $\biat_i
  (A \with B)$. In such a way that $\biat_i
  (A \with B)$ implies $A$ and implies $B$, but does not imply $A\otimes B$.
\item Bringing about the transformation of the resource $A$ into the
  resource $B$:\linebreak $\biat_i (A \multimap B)$. In such as way
  that $A \otimes \biat_i (A \multimap B)$ implies $B$.
\end{itemize}

\medskip
The only structural rule that holds in Linear Logic is exchange (E),
that is responsible for the commutativity of the multiplicative operators and amounts
to forgetting sequentiality of actions, e.g., $A \otimes B
\rightarrow B \otimes A$. We will see in this paper how
to deal with ordered information. We will still admit exchange for the
two conjunctions $\with$ and $\otimes$ (commutative conjunctions), but
we will introduce a non-commutative counterpart $\odot$ to the
multiplicative $\otimes$. The formula $A \odot B$ is not equivalent
to $B \odot A$. Since $\odot$ is non-commutative, we can have two
\emph{order-sensitive} linear implications (noted in \cite{lambek1958}
$\setminus$ and $/$).

For agency in multi-agent systems, resources must often become
available at key points in a series of transformations. The order of
resource production and transformation becomes quickly relevant. It is
then interesting to talk about:
\begin{itemize}
\item Bringing about the resource $A$ first, then $B$: $\biat_i (A \odot B)$. In such a way that $\biat_i (A \odot B)$ is not equivalent to $\biat_i (B \odot A)$.
\item Bringing about the order-sensitive transformation of the resource from $A$ into $B$: $\biat_i(A \setminus B)$. In such a way that $A \odot \biat_i(A \setminus B)$ implies $B$, but $\biat_i(A \setminus B) \odot A$ does not.
\end{itemize}

These considerations motivate us to investigate the theoretical underpinnings
of modal versions of resource-conscious logics with the listed
propositional operators.

We will first concentrate on extensions with \emph{one}
\emph{minimal} \emph{non-normal} modality. In a second part,
where we address modalities of agency, we will exploit our results
that will naturally scale up to logics with
\emph{\smash{multiple}} \emph{non-minimal} (but still
non-normal) modalities. 

\medskip\noindent\textbf{Outline.}~
In Section~\ref{sec:modal-krm}, we present the Kripke resource models
which already exist in the literature. The semantics of all the
languages studied in this paper will be adequate extensions of Kripke
resource models. We first enrich the Kripke resource models with
neighbourhood functions to capture non-normal modalities. We obtain
what we simply coin modal Kripke resource models. We define and study
a minimal non-normal modal logic, \MILL. We introduce a Hilbert system
in Section~\ref{sec:hilbert} and a sequent calculus in
Section~\ref{sec:sequent-mill}; Both are shown sound and
complete. Moreover, we can easily show that the sequent calculus admits cut
elimination that provides a normal form for proofs. Proof search is
proved to be in PSPACE. 

We extend our framework in Section~\ref{sec:mill-nc} to account
for partially commutative Linear Logic that allows for integrating
commutative and non-commutative operators.

Then in Section~\ref{sec:linear-biat} we instantiate the minimal modal
logic with a resource-sensitive version of the logics of
bringing-it-about: \RSBIAT. Again, sound and complete Hilbert system
and sequent calculus are provided. Proof-search in \RSBIAT is
PSPACE-easy. We also show how \RSBIAT can be extended with the
non-commutative language and thus represent sequentiality of actions.
In Section~\ref{sec:application}, we motivate and discuss a number of
applications of our system to represent and reason about artefacts.

\section{\MILL and modal Kripke resource models}
\label{sec:modal-krm}

Let $Atom$ be a non-empty set of atomic propositions. We introduce the most basic language studied in this paper. 
The language $\mathcal{L}_{\MILL}$ is given by the BNF:
\[A ::= \textbf{1} \mid p \mid A \otimes A \mid A \with A \mid A \implies A \mid \Box A\] where $p \in Atom$.  
It is a modal version of what corresponds to the language of
propositional intuitionistic Linear Logic, but without the additive
disjunction and the additive units. This propositional part can also be seen as the fragment
of BI~\cite{PymBI1999} without additive disjunction and implication.

Let us first concentrate on the propositional part for which a
semantics already exists in the literature. We call the logic \ILL and
$\mathcal{L}_{\ILL}$ its language. A Kripke-like class of models for
\ILL is basically due to Urquhart \cite{UrquhartJSL1972}. A
\emph{Kripke resource frame} is a structure $\mathcal{M} = {(M, e,
  \circ, \geq)}$, where $(M,e,\circ)$ is a commutative monoid with
neutral element $e$, and $\geq$ is a pre-order on $M$.  The frame has
to satisfy the condition of \emph{bifunctoriality}: if $m \geq n$, and
$m' \geq n'$, then $m \circ m' \geq n \circ n'$. To obtain a
\emph{Kripke resource model}, a valuation on atoms $V: Atom \to
\mathcal{P}(M)$ is added. It has to satisfy the \emph{heredity}
condition: if $m \in V(p)$ and $n \geq m$ then $n \in V(p)$.

The truth conditions of $\mathcal{L}_{\ILL}$ in the Kripke resource
model $\mathcal{M} = {(M, e, \circ, \geq, V)}$ of the formulas of the
propositional part are the following:
\begin{description}
\item $m \models_\mathcal{M} p$ iff $m \in V(p)$.
\item $m \models_\mathcal{M} \textbf{1}$ iff $m \geq e$.
\item $m \models_\mathcal{M} A \tensor B$ iff there exist $m_{1}$ and $m_{2}$ such that $m \geq m_{1} \circ m_{2}$ and $m_{1} \models_\mathcal{M} A$ and  
$m_{2} \models_\mathcal{M} B$.
\item $m \models_\mathcal{M} A \with B$ iff $m \models_\mathcal{M} A$ and $m \models_\mathcal{M} B$.
\item $m \models_\mathcal{M} A \implies B$ iff for all $n \in M$, if $n \models_\mathcal{M} A$, then $n \circ m \models_\mathcal{M} B$.  
\end{description}

Observe that heredity can be shown to extend naturally to every
formula, in the sense that:
\begin{prop}\label{prop:propositional-heredity}
For every formula $A \in \mathcal{L}_{\ILL}$, if $m \models A$ and $m' \geq m$, then $m'
\models A$.
\end{prop}

\medskip\noindent\textbf{Notations.}~
An intuitionistic negation can be added to the language. We simply
choose a designated atom $\bot \in Atom$ and decide to conventionally
interpret $\bot$ as indicating a contradiction. Negation is then
defined by means of implication as $\sim A \equiv A \implies \bot$
\cite{KanovichEtAlMSCT2006}: the occurrence of $A$ yields the
contradiction. There will be no specific rule for negation.

Given a multiset of formulas, it will be useful to combine them into a
unique formula. We adopt the following notation: $\emptyset^* =
\textbf{1}$, and $\Delta^* = A_1 \otimes \dots \otimes A_k$ when $\Delta
= \{A_1, \ldots ,A_k\}$. 

Denote $||A||^\mathcal{M}$ the extension of $A$ in $\mathcal{M}$,
i.e.~the set of worlds of $\mathcal{M}$ in which $A$ holds. A formula
$A$ is \emph{true} in a model $\mathcal{M}$ if $e
\models_\mathcal{M} A$.\footnote{When no confusion can arise we will write $||A||$ instead of $||A||^\mathcal{M}$, and $m \models A$ instead of $m \models_\mathcal{M} A$.} A formula $A$ is \emph{valid} in Kripke resource frames, noted $\models A$, iff it is true in every model.

\medskip\noindent\textbf{Modal Kripke resource models.}~
Now, to give a meaning to the modality, we define a neighbourhood
semantics on top of the Kripke resource frame. A neighbourhood function
is a mapping $N: M \to \mathcal{P}(\mathcal{P}(M))$ that associates a
world $m$ with a set of sets of worlds. (See~\cite{chellas80}.) We
define:
\begin{description}
\item $m \models \Box A$ iff $|| A || \in N(m)$
\end{description}
This is not enough, though. It is possible that $m \models \Box A$,
yet $m' \not \models \Box A$ for some $m' \geq m$. That is,
Proposition~\ref{prop:propositional-heredity} does not hold with the
simple extension of $\models$ for $\mathcal{L}_{\MILL}$. (One disastrous
consequence is that the resulting logic does not satisfy the \emph{modus
ponens} or the cut rule.) We could define the clause concerning the
modality alternatively as: $m \models \Box A$ iff there is $n \in M$,
such that $m \geq n$ and $|| A || \in N(n)$. However, it will be more agreeable to keep working with the standard definition and instead impose a condition on the models.

We will require our neighbourhood function to satisfy the condition
that if some set $X \subseteq M$ is in the neighbourhood of a world,
then $X$ is also in the neighbourhood of all ``greater''
worlds.\footnote{An analogous yet less transparent condition was used
  in~\cite{DagostinoEtAlSL1997} for a normal modality.} Formally, our
modal Linear Logic is evaluated over the following models:
\begin{definition}\label{def:mkrm}
A \emph{modal Kripke resource model} is a structure $\mathcal{M} = (M,
e, \circ, \geq, N, V)$ such that:
\begin{itemize}
\item $(M, e, \circ, \geq)$ is a Kripke resource frame;
\item $N$ is a neighbourhood function such that:
\begin{equation}
\label{princ:heredityN}
\text{if } X \in N(m) \text{ and } n \geq m \text{ then } X \in N(n)
\end{equation}
\end{itemize}
\end{definition}
It is readily checked that
Proposition~\ref{prop:propositional-heredity} is true as well for
$\mathcal{L}_{\MILL}$ over modal Kripke resource models for modal
formulas.
We thus have:
\begin{prop}\label{prop:modal-heredity}
For every formula $A \in \mathcal{L}_{\MILL}$, if $m \models A$ and
$m' \geq m$, then $m' \models A$.
\end{prop}
\begin{proof}
\dnote{We handle the new case $A = \Box B$. Suppose $m \models \Box B$ and
$m'\geq m$. By definition, $||B|| \in N(m)$, and thus $||B|| \in N(m')$
follows from condition (\ref{princ:heredityN}). Therefore $m' \models
  \Box B$.}
\end{proof}

\section{Hilbert system for \MILL and soundness}
\label{sec:hilbert}

\begin{table}[h]
\begin{center}
\framebox{
\begin{tabular}{l}
$A \implies A$\\
$(A \implies B) \implies ((B \implies C) \implies (A \implies C))$\\
$(A \implies (B \implies C)) \implies (B \implies (A \implies C))$\\
$A \implies (B \implies A \otimes B)$\\
$(A \implies (B \implies C)) \implies (A \otimes B \implies C)$\\
$\textbf{1}$\\
$\textbf{1} \implies (A \implies A)$\\
$(A \with B) \implies A$\\
$(A \with B) \implies B$\\
$((A \implies B) \with (A \implies C)) \implies (A \implies B \with C)$\\\\
\end{tabular} }
\end{center}
\caption{\label{table-hmill} Axiom schemata in \HMILL}
\end{table}





This part extends the Hilbert system for \ILL from \cite{Troelstra1992} and
\cite{AvronTCS1988}.
We define the Hilbert-style calculus $\HMILL$ for \MILL by defining the following notion of deduction. 
\begin{definition}[Deduction in \HMILL]\label{def:derivation-hmill}
A \emph{deduction tree} in \HMILL  $\mathcal{D}$ is inductively constructed as follows. (i)~The leaves of the tree are assumptions $A \t_\H A $, for $A \in \mathcal{L}_\MILL$, or
$\t_\H B$ where $B$ is an axiom in Table~\ref{table-hmill} (base cases).\\
 (ii) We denote by $\stackrel{\mathcal{D}}{\Gamma \t_\H A}$ a deduction tree with conclusion $\Gamma \t_\H A$. If $\mathcal{D}$ and $\mathcal{D}'$ are deduction trees, then the following are deduction trees (inductive steps). 

\begin{center}
\begin{tabular}{cc}
$
\AC{\stackrel{\mathcal{D}}{\Gamma \t_\H A}}
\AC{\stackrel{\mathcal{D}'}{\Gamma' \t_\H A \implies B}} \rl{$\implies$-rule}
 \BC{\Gamma, \Gamma' \t_\H  B}
 \dip
$

& 

$
\AC{\stackrel{\mathcal{D}}{\Gamma \t_\H A}}
\AC{\stackrel{\mathcal{D}'}{\Gamma \t_\H B}} \rl{$\with$-rule}
\BC{\Gamma \t_\H A \with B}
\dip
$
\\
\end{tabular}
$$
 \AC{\stackrel{\mathcal{D}}{\vdash_\H A \implies B}}
 \AC{\stackrel{\mathcal{D}'}{\vdash_\H B \implies A}}
 \rl{$\Box$(re)}
 \BC{\t_\H \Box A \implies \Box B}
 \dip
$$
\end{center}

\end{definition}
We sometimes refer to the $\implies$-rule as \emph{modus ponens}. 
We say that $A$ is deducible from $\Gamma$ in \HMILL and we write $\Gamma \t_{\HMILL} A$ iff there exists a deduction tree in \HMILL with conclusion $\Gamma \t_\H A$. The deduction without assumptions, i.e. $\t_\HMILL A$, is just a special case of the above definition.

In general when defining Hilbert systems for linear logics, we need to
be careful in the definition of derivation from assumptions: since
$\Gamma$ is a multiset, we need to handle occurrences of hypothesis
when applying for instance modus ponens.  From
Definition~\ref{def:derivation-hmill}, every occurrence of assumptions or axioms in a derivation, except for
the conclusion, is used exactly once by an application of modus
ponens~\cite{AvronTCS1988}.  With respect to this notion of
derivation, the deduction theorem holds.
 \begin{theorem}[Deduction theorem for \HMILL]\label{thm:deduction} If ${\Gamma, A \vdash_\HMILL B}$ then\linebreak ${\Gamma \vdash_\HMILL A \implies B}$.  
 \end{theorem}
\begin{proof}
The deduction theorem holds for the propositional fragment \ILL, see~\cite[pp.~66-68]{Troelstra1992}. The only case to consider is the rule $\Box$(re). However this trivially holds because the contexts of the sequents in the rule are empty.
\end{proof}

\begin{remark}\label{rem:rulesre} We defined the rule $\Box$(re) in that particular way because of the deduction theorem. With the above formulation, the deduction theorem is preserved in \MILL.
Moreover, by our definition of deduction tree, this version entails the rule $\Box$(re'): from two deductions trees with conclusion $A \vdash_\H B$ and $B \vdash_\H A$, we can build a deduction tree with conclusion $\Box A \vdash_\H \Box B$.


We claim that $\Box$(re) implies $\Box$(re'). 
Assume the premises of $\Box$(re'): there are two deduction trees with
conclusions $A \vdash_\H B$ and $B\vdash_\H A$.  In virtue of
Theorem~\ref{thm:deduction} and the definition of $\vdash_\HMILL$, we
know there are two trees with conclusions $\vdash_\H A \implies B$ and
$\vdash_H B \implies A$. Thus by $\Box$(re), we can build a deduction
tree with conclusion $\vdash_\H \Box A \implies \Box B$. Together with
$\Box A \vdash_\H \Box A$ (a leaf), the $\implies$-rule gives us a
deduction tree with conclusion $\Box A \vdash_\H \Box B$.

With the version $\Box$(re'), the deduction theorem fails, as we do not
have any axiom that talks about the modality in this case.
\end{remark}

We can prove the soundness of \HMILL wrt.~our semantics.
\begin{theorem}[Soundness of \HMILL]\label{thm:soundness-hmill} If ${\Gamma \vdash_{\HMILL} A}$ then, for every model, $\Gamma \models A$ (namely, $e \models (\Gamma)^{*} \implies A$).
\end{theorem}
\begin{proof}

We only give the arguments of the proof of soundness for two representative cases.\\

\emph{Soundness of $\implies$-rule}. We show now that $\implies$-rule preserves validity. Namely, we prove, by induction on the length of the derivation tree that 
if (1)~$e \models \Gamma \implies  A$ and (2)~$e \models \Gamma' \implies (A \implies B)$, then $e \models \Gamma \otimes \Gamma' \implies B$.
The first assumption entails that for all $x$, if $x \models \Gamma$, then $x \models A$. 
The second assumption entails that for all $y$, if $y \models \Gamma'$, then $y \models A \implies B$. Thus, for all $t$, if $t \models A$, then $y \circ t \models B$. \par
Let $z \models \Gamma \otimes \Gamma'$, thus there exist $z_{1}$ and $z_{2}$ such that $z_{1} \models \Gamma$, $z_{2} \models \Gamma'$ and 
$z \geq z_{1} \circ z_{2}$. By (1), we have that $z_{1} \models A$ and, by (2), $z_{2} \models A \implies B$, thus $z_{1} \circ z_{2} \models B$, and by Proposition~\ref{prop:modal-heredity} we have that $z \models B$.\footnote{Remember that Proposition~\ref{prop:modal-heredity} holds for the modal language $\mathcal{L}_{\MILL}$ because of Condition~(\ref{princ:heredityN}) on modal Kripke resource models.}
\\

\noindent
\emph{Soundness of $\Box$(re)}. We show that $\Box$(re) preserves  validity, namely, if $e \models A \implies B$ and $e \models B \implies A$, then $e \models \Box A \implies \Box B$. Our assumptions imply that, for all $x$, if $x \models A$, then $x \models B$, and if $x \models B$ then $x \models A$. Thus, 
$||A|| = ||B||$. We need to show that for all $x$, if $x \models \Box A$, then $x \models \Box B$. 
By definition,  $x \models \Box A$ iff $||A|| \in N(x)$. Thus, since $||A|| = ||B||$, we have that $||B|| \in N(x)$, that means  $x \models \Box B$.
\end{proof}

\section{Sequent calculus \MILL and completeness}
\label{sec:sequent-mill}

In this section, we introduce the sequent calculus for our logic.  A \emph{sequent} is a statement $ \Gamma \t A$ where $\Gamma$ is a finite multiset of occurrences of formulas of \MILL and $A$ is a formula. The fact that we allow for a single formula in the conclusions of the sequent corresponds to the fact that we are working with the intuitionistic version of the calculus~\cite{Girard1987}.    

\begin{table}[h]
\small
\begin{center}

\begin{tabular}[]{cc}
\AxiomC{} \RightLabel{\footnotesize ax} \UnaryInfC{$A \vdash A$} \DisplayProof   & \AxiomC{$\Gamma, A
\vdash C$} \AxiomC{$\Gamma' \vdash A $} \RightLabel{\footnotesize cut} \BinaryInfC{$\Gamma,
\Gamma' \vdash C$} \DisplayProof
\\
\end{tabular}
\begin{tabular}[]{cc}
&  \\
\AxiomC{$\Gamma, A, B \vdash C$} \RightLabel{$\otimes$L} \UnaryInfC{$\Gamma, A \otimes B \vdash
C$} \DisplayProof  &  \AxiomC{$\Gamma \vdash A $} \AxiomC{$\Gamma' \vdash B$}
 \RightLabel{$\otimes$R} \BinaryInfC{$\Gamma, \Gamma' \vdash A\otimes B$} \DisplayProof      \\
& \\
\AxiomC{$\Gamma \vdash A$} \AxiomC{$\Gamma', B \vdash C$} \RightLabel{$\multimap$L}
\BinaryInfC{$\Gamma', \Gamma, A\multimap B \vdash C$} \DisplayProof  &

\AxiomC{$\Gamma, A \vdash B$} \RightLabel{$\multimap$R}
\UnaryInfC{$\Gamma \vdash A \multimap B$}   \DisplayProof \\
& \\
\AxiomC{$\Gamma, A, \Gamma' \vdash C$} \RightLabel{$\with$L} \UnaryInfC{$\Gamma, A \with B, \Gamma' \vdash
C$} \DisplayProof  &    \AxiomC{$\Gamma, B, \Gamma' \vdash C$} \RightLabel{$\with$L} \UnaryInfC{$\Gamma, A \with B, \Gamma' \vdash
C$} \DisplayProof    \\
& \\
\end{tabular}

\begin{tabular}[]{ccc}
\AxiomC{$\Gamma \vdash A $} \AxiomC{$\Gamma \vdash B$}
 \RightLabel{$\with$R} \BinaryInfC{$\Gamma \vdash A\with B$} \DisplayProof & \AC{\Gamma \t C} \rl{$\textbf{1}$L} \UC{\Gamma, \textbf{1} \t C} \DisplayProof &
\AC{} \rl{$\textbf{1}$R} \UC{\t \textbf{1}} \DisplayProof  
\\
\end{tabular}
\end{center}
\begin{center}
\AC{A \t B}
\AC{B \t A}\rl{$\Box$(re)}
\BC{\Box A \t \Box B}
\dip
\end{center}

\caption{Sequent calculus \MILL \label{table:mill}}
\end{table}

Since in a sequent $\Gamma \vdash A$ we identify $\Gamma$ to a
multiset of formulas, the exchange rule---the reshuffling of
$\Gamma$---is implicit.

A sequent $\Gamma \t A$ where $\Gamma = A_{1}, \ldots,
A_{n}$ is \emph{valid} in a modal Kripke resource frame iff the
formula $A_{1} \tensor \dots \tensor A_{n} \implies A$ is valid, namely $\models \Gamma^{*} \implies A$. 


We obtain the sequent calculus for our minimal modal logic \MILL by extending the language of \ILL with modal formulas and by adding a new rule. Saturating the notation, we label the sequent calculus rule like the rule for equivalents in the Hilbert system: $\Box$(re). The calculus is shown in Table~\ref{table:mill}.
 
To establish a link with the previous section, we first show that provability in sequent calculus is equivalent to provability in the Hilbert system. 
\begin{theorem}\label{sequenthilbert} It holds that $\Gamma \t_{\HMILL} A$ iff the sequent $\Gamma \t A$ is derivable in the sequent calculus for \MILL. 
\end{theorem}

\begin{proof}  
(Sketch) The propositional cases are proved in~\cite{AvronTCS1988}, we
only need to extend it for the case of the modal rules. This is done
by induction on the length of the proof. For instance, in one
direction, assume that a derivation $\mathcal{D}$ of $\vdash_H \Box
A \implies \Box B$ is obtained by $\mathcal{D'}$ of $\vdash_H A \implies
B$ and $\mathcal{D''}$ of $\vdash_H B \implies A$. 
By definition of deduction and of $\vdash_\HMILL$,
we have that $A \vdash_{\HMILL} B$ and
$B \vdash_{\HMILL} A$. By induction hypothesis, we have that $A \vdash
B$ and $B \vdash A$ are provable in the sequent calculus. Thus, by
$\Box$(re) and $\implies R$, we have that $\t \Box A \implies \Box B$
is provable in the sequent calculus.
\end{proof}

Crucially, the modal extension does not affect cut elimination. Cut elimination holds for Linear Logic \cite{Girard1987}.  The proof for \MILL largely adapts the proof for Linear Logic \cite{Troelstra1992}. Recall that the \emph{rank} of a cut is the complexity of the cut formula. The \emph{cutrank} of a proof is the maximum of the ranks of the cuts in the proof. The \emph{level} of a cut is the length of the subproof ending in the cut, \cite{Troelstra1992}. 
The proof of cut elimination proceeds by induction on the cutrank of the proof. It is enough to assume that the occurrence of the cut with maximal rank is the last rule of the proof. The inductive step shows how to replace a proof ending in a cut with rank $n$  with a proof with the same conclusion and smaller cutrank. The proof of the inductive step proceeds by induction on the level of the terminal cut, namely on the length of the proof. 
Note that if one of the premises of a cut is an axiom, then we can simply eliminate the cut.
Given a sequent rule $R$, the occurrence of a formula $A$ in the conclusion of $R$ is principal if $A$ has been introduced by $R$.

\begin{theorem}\label{thm:cut-elim} Cut elimination holds for \MILL.
\end{theorem}

\begin{proof}(Sketch) 
As usual, there are two main cases to consider: first, the cut formula is principal in both premises of the terminal cut rule, and second, that is not the case.

\noindent
\emph{First main case}.  We replace the cut of maximal rank with two cuts with strictly smaller rank. 
\
For example, take the case in which $\Box C$ is the cut formula and is principal in both premises (i.e. it has been introduced by $\Box$(re)):  
\[
\AC{B \t C}
\AC{C \t B}
\rl{$\Box$(re)}
\BC{\Box B \t \Box C}
\AC{C \t D}
\AC{D \t C}
\rl{$\Box$(re)}
\BC{\Box C \t \Box D}
\rl{cut}
\BC{\Box B \t \Box D}
\dip
\]
It is reduced by replacing the cut on $\Box C$ by two cuts on $C$ with strictly smaller rank.  
\[
\AC{B \t C}
\AC{C \t D}
\rl{cut}
\BC{B \t D}
\AC{D \t C}
\AC{C \t B}
\rl{cut}
\BC{D \t B}
\rl{$\Box$(re)}
\BC{\Box B \t \Box D}
\dip
\]
\emph{Second main case}. The cut formula is not principal in one of the premises of the cut rule. Suppose $R$ is the rule that does not introduce the cut formula. In this case, we can apply the cut after $R$. By induction, the proof minus $R$ can be turned into a cut-free proof, since the length of the subproof is smaller and the cutrank is equal or smaller.  By applying again $R$ to the cut-free proof,  we obtain a proof with the same conclusion of the starting proof and less cuts.

For example,

\[
\AC{B \t C}
\AC{C \t B}
\rl{$\Box$(re)}
\BC{\Box B \t \Box C}
\AC{...}
\rl{R}
\UC{\Gamma, \Box C \t A}
\rl{cut}
\BC{\Box B, \Gamma \t A}
\dip
\]

can be turned into:

\[
\AC{B \t C}
\AC{C \t B}
\rl{$\Box$(re)}
\BC{\Box B \t \Box C}
\AC{...}
\rl{}
\UC{\Gamma', \Box C \t A'}
\rl{cut}
\BC{\Box B, \Gamma' \t A'}
\AC{\Gamma'' \t A''}
\rl{R}
\BC{\Box B, \Gamma \t A}
\dip
\]

\end{proof}

By inspecting the rules others than cut, it is easy to see that cut elimination entails the subformula property, namely if $\Gamma \t A$ is derivable,
then there is a derivation containing subformulas of $\Gamma$ and $A$ only. 

The decidability remains to be established. We can show that the
proof-search for \MILL is no more costly in terms of space than the proof-search for
propositional intuitionistic multiplicative additive L.inear
Logic~\cite{LincolnMitchellScedrov92}.
\begin{theorem}\label{thm:PSPACE} Proof search complexity for \MILL is in PSPACE.
\end{theorem}

\begin{proof} (Sketch) The proof adapts the argument in~\cite{LincolnMitchellScedrov92}. By cut elimination, Theorem~\ref{thm:cut-elim}, for every provable sequent in \MILL there is a cut-free proof with same conclusion. 
For every rule in \MILL other than (cut), the premises have a strictly
lower complexity wrt.~the conclusion. Hence, for every provable
sequent, there is a proof whose branches have a depth at most linear
in the size of the sequent. The size of a branch is at most quadratic
in the size of the conclusion. And it contains only subformulas of the
conclusion sequent because of the subformula property. This means that
one can non-deterministically guess such a proof, and check each
branch one by one using only a polynomial space. Proof search is then
in NPSPACE = PSPACE.
\end{proof}

We present the proof of completeness of \MILL wrt.~the class of modal Kripke resource frames. 

\begin{theorem}[Completeness of the sequent calculus]\label{thm:compl-sc}
If $\models \Gamma^* \implies \A$ then $\Gamma \vdash \A$.
\end{theorem}



The proof can be summarised as follow. We build a canonical model
$\mathcal{M}^c$~(Definition~\ref{def:mc}). In particular, the set
$M^c$ of states consists in the set of finite multisets of formulas, and the
neutral element $e^c$ is the empty multiset. We first need to show
that it is indeed a modal Kripke resource model~(Lemma~\ref{lem:mc}). Second
we need to show a correspondence, the ``Truth Lemma'', between
$\vdash$ and truth in $\mathcal{M}^c$. Precisely we show that for a
formula $\A$ and a multiset of formulas $\Gamma \in M^c$, it is the
case that $\Gamma$ satisfies $\A$ iff $\Gamma \vdash \A$ is
provable in the calculus~(Lemma~\ref{lem:truth}). Finally, to show
completeness, assume that it is not the case that $\vdash \Gamma^*
\implies \A$. By the Truth Lemma, it means that in the canonical
model $\Gamma^* \implies \A$ is not satisfied at $e^c$. So
$\mathcal{M}^c$ does not satisfy $\Gamma^* \implies \A$. So it is
not the case that $\models \Gamma^* \implies \A$.

\medskip

We construct the canonical model $\mathcal{M}^c$, then we prove that $\mathcal{M}^c$ is a modal Kripke resource model, and we prove the Truth
Lemma. 

In the following, $\sqcup$ is the multiset union. Also, $\mid \A
\mid^c = \{ \Gamma \mid \Gamma \vdash \A \}$.
\begin{definition}\label{def:mc}
Let $\mathcal{M}^c = (M^c,  e^c, \circ^c, \geq^c, N^c, V^c)$ such that:
\begin{itemize}
\item $M^c = \{ \Gamma \mid \Gamma \text{ is a finite multiset of formulas} \}$;
\item $\Gamma \circ^c \Delta = \Gamma \sqcup \Delta$;
\item $e^c = \emptyset$;
\item $\Gamma \geq^c \Delta$ iff $\Gamma \vdash \Delta^*$;
\item $\Gamma \in V^c(p)$ iff $\Gamma \vdash p$;
\item $N^c(\Gamma) = \{ \mid \A \mid^c \mid \Gamma \vdash \Box\A\}$.
\end{itemize}
\end{definition}

\begin{lemma}\label{lem:mc}
$\mathcal{M}^c$ is a modal Kripke resource model.
\end{lemma}
\begin{proof}
\noindent
1. $(M^c, e^c, \circ^c, \geq^c)$ is the ``right type'' of ordered monoid:
(i)~$(M^c,e^c,\circ^c)$ is a commutative monoid with neutral element
$e^c$, and (ii)~$\geq^c$ is a pre-order on $M^c$.  Finally, (iii)~if
$\Gamma \geq^c \Delta$ and $\Gamma' \geq^c \Delta'$ then $\Gamma
\circ^c \Gamma' \geq^c \Delta \circ^c \Delta'$.

For~(i), commutativity (and associativity) follow from the definition
of $\circ^c$ as the multiset union, and the neutrality of $e^c$
follows from it being the empty multiset---the neutral element of the
multiset union.

For~(ii), $\geq^c$ is reflexive because $\{A_1, \ldots, A_n\} \vdash
\{A_1, \ldots, A_n\}^*$ can be proved from the axioms (ax) $A_k \vdash
A_k$, $1 \leq k \leq n$, and by applying $\otimes$R $n-1$ times. The
key rule to establish that $\geq^c$ is transitive is $cut$. 

For~(iii), assume $\Gamma \geq^c \Delta$ and $\Gamma' \geq^c \Delta'$,
that is, $\Gamma \vdash \Delta^*$ and $\Gamma' \vdash
\Delta'^*$. By $\otimes$R we have $\Gamma, \Gamma' \vdash
\Delta^* \otimes \Delta'^*$. By applying the definitions we end
up with $\Gamma \sqcup \Gamma' \vdash (\Delta \sqcup \Delta')^*$
and the expected result follows.\\

\noindent
2. $V^c$ is a valuation function and satisfies heredity: if $\Gamma
\in V(p)$ and $\Delta \geq^c \Gamma$ then $\Delta \in V(p)$. To see
this, suppose $\Gamma \vdash p$ and $\Delta \vdash \Gamma^*$. By
applying $\otimes$L enough times, we have $\Gamma^* \vdash p$. By
$cut$, we obtain $\Delta \vdash p$.\\

\noindent
3. $N^c$ is well-defined: Suppose that $\mid \A \mid^c =
\mid \B \mid^c$. We need to show that $\mid \A \mid^c \in
N^c(\Gamma)$ iff $\mid \B \mid^c \in N^c(\Gamma)$.\\

From $\mid \A \mid^c = \mid \B \mid^c$, we have $\Gamma \vdash \A
\Rightarrow \Gamma \vdash \B$. In particular, we have $\A \vdash \A
\Rightarrow \A \vdash \B$. Hence, $\A \vdash \B$ is provable (by
rule~(ax)). We show symmetrically that $\B \vdash \A$ is provable.\\

From $\A \vdash \B$ and $\B \vdash \A$, we have by rule
$\Box$(re) that $\Box \A \vdash \Box\B$ is provable, and also that
$\Box \B \vdash \Box\A$ is provable.

Now suppose that $\Gamma \vdash \Box\A$. Since $\Box \A \vdash
\Box\B$ is provable, we obtain by $cut$ that $\Gamma \vdash
\Box\B$ is provable. Symmetrically, suppose that $\Gamma \vdash
\Box\B$. Since $\Box \B \vdash \Box\A$ is provable, we obtain by
$cut$ that $\Gamma \vdash \Box\A$ is provable. 

Hence, we have that $\Gamma \vdash \Box \A$ iff $\Gamma \vdash \Box
\B$. By definition of $N^c$, it means that $\mid \A \mid^c \in
N^c(\Gamma)$ iff $\mid \B \mid^c \in N^c(\Gamma)$.\\ 

\noindent
4. If $X \in N^c(\Gamma)$ and $\Delta \geq^c \Gamma$ then $X \in
N^c(\Delta)$. To see that this is the case, the hypotheses are
equivalent to $\Gamma \vdash \Box A$ for some $A$ such that $\mid A
\mid^c = X$, and $\Delta \vdash \Gamma^*$. By repeatedly applying
$\otimes$L to obtain $\Gamma^* \vdash \Box A$ and by using $cut$,
we infer that $\Delta \vdash \Box A$. Which is equivalent to the
statement that $X \in N^c(\Delta)$.
\end{proof}

Let us then denote by $\models_c$ the truth relation in $\mathcal{M}^c$.

\begin{lemma}\label{lem:truth}
$\Gamma \models_c \A$ iff $\Gamma \vdash \A$.
\end{lemma}
\begin{proof}
By structural induction on the form of $\A$. For the base case, we have for every atom $p$ that $\Gamma \models_c p$ iff $\Gamma \in V^c(p)$ iff $\Gamma \vdash p$.
For induction, we suppose that the lemma holds for a formula $\B$ (Induction Hypothesis). The cases of the propositional connectives are found in~\cite{kamide03bulletin}. We prove here the case $\A = \Box \B$.






We have the following sequence of equivalences of $\Gamma
\models_c \Box \B$:
\begin{itemize}
\item[] iff $\mid\mid\B\mid\mid^{\mathcal{M}^c} \in N^c(\Gamma)$, by
  definition of $\models_c$;
\item[] iff $\{ \Delta \mid \Delta \models_c \B \} \in N^c(\Gamma)$,
  by definition of $\mid\mid . \mid\mid^{\mathcal{M}^c}$;
\item[] iff $\{ \Delta \mid \Delta \vdash \B \} \in N^c(\Gamma)$, by
  Induction Hypothesis;
\item[] iff $\mid\B\mid^c \in N^c(\Gamma)$, by definition of $\mid . \mid^c$;
\item[] iff $\Gamma \vdash \Box \B$, by definition of $N^c$.
\end{itemize}
\end{proof}

\medskip

We could now prove that the sequent calculus is sound, and we could
adapt our proof of Theorem~\ref{thm:compl-sc} to prove that the Hilbert
system \HMILL is complete. But we are already there; We have the
following:
\begin{itemize}
\item if $\Gamma \vdash_{\HMILL} A$ then $e \models \Gamma^* \implies
  A$ (Theorem~\ref{thm:soundness-hmill});
\item if $e \models \Gamma^* \implies A$ then $\Gamma \vdash A$
  (Theorem~\ref{thm:compl-sc});
\item if $\Gamma \vdash A$ then $\Gamma \vdash_{\HMILL} A$ (Theorem
  \ref{sequenthilbert}).
\end{itemize}
Therefore, the completeness of the Hilbert system and the soundness of the
sequent calculus both follow.
\begin{corollary}
We have:
\begin{itemize}
\item if  ${e \models \Gamma^* \implies
  A}$ { then } ${\Gamma \vdash_{\HMILL} A}$;
\item if  ${\Gamma \vdash A}$  { then } ${e \models \Gamma^* \implies A}$.
\end{itemize}
\end{corollary}

\section{Adding non-commutativity}
\label{sec:mill-nc}

Systems that integrate a commutative Linear Logic with a non-commutative one have been studied in \cite{DeGroote1996,AbrusciRuet2000,Retore1997}. Note that the purely non-commutative version of intuitionistic Linear Logic is basically the calculus with two order sensitive implications developed by Lambek~\cite{lambek1958}.
The basic propositional logic that we use is provided by \cite{DeGroote1996} and labelled PCL, partially commutative Linear Logic. The main novelty is that the structural rule of exchange no longer holds in general. The context of a sequent is now only partially commutative, and is now built by means of two constructors. Thus, we essentially use \emph{context} as a shorthand for \emph{partially commutative context}. Every formula (to be defined shortly after) is a context. Then, for every context $\Gamma$ and $\Delta$, we can build their \emph{parallel} composition $(\Gamma , \Delta)$ (primarily commutative context) and their \emph{serial} composition $(\Gamma ; \Delta)$ (primarily non-commutative context) \cite{DeGrooteRetoreSPO1997}. A context can thus be seen as finite tree with non-leaf nodes labelled with `$;$' or `$,$' and with leafs labelled by  formulas. Two branches emanating from a `$,$' commute with each other, while the branches emanating from a `$;$' node do not.
We write $()$ for the empty context, and assume that it acts as the identity element of the parallel and serial composition, that is: $((),\Gamma) = (\Gamma, ()) = (();\Gamma) = (\Gamma; ()) =
\Gamma$.
We denote by $\Gamma[-]$ a context ``with a hole'', and $\Gamma[\Delta]$ denotes this very context with the ``hole filled'' with the context $\Delta$.

The language of \PCMILL extends the language of \MILL by adding the following operators: the non-commutative tensor noted $\odot$ and the two order sensitive implications noted $\setminus$ and $/$:
\[A ::= \textbf{1} \mid p \mid A \otimes A \mid A \with A \mid A \implies A \mid A \odot A \mid A \setminus A \mid A / A \mid \Box
A\] where $p \in Atom$.

In order to blend together commutative and non-commutative sequence, we have to choose what is the interpretation of commutativity, namely if we view a commutative concurrent process such as $A \otimes B$ as entailing that either directions are allowed \cite{DeGroote1996,DeGrooteRetoreSPO1997}. That is, parallel composition is weaker than serial composition, so it shall hold that $A \otimes B \implies A \odot B$. This means that if two resources can be combined with no particular order, then they can be combined sequentially. This choice is reflected by the structural rule of \emph{entropy} (ent) below. 
The first two lines of Table~\ref{table:PCMILL} state the associativity of serial and parallel compositions, the third line states the commutativity of parallel composition and the entropy principle. 
\dnote{The meaning of the two implications can be expressed in terms of \emph{pre-conditions} and \emph{post-conditions}. $A \setminus B$ requires that $A$ occurs \emph{before} (i.e. on the left of) the implication: $A; A \setminus B \t B$. By contrast, $B / A$ requires that $A$ occurs \emph{after} the implication, accordingly: $B / A; A \t B$.}

\begin{table}[h]
\begin{center}
\emph{Structural rules}

\begin{tabular}[]{cc}
&  \\
\AC{\Gamma[\Delta_1, (\Delta_2, \Delta_3)] \t A}
\rl{,a1}
\UC{\Gamma[(\Delta_1, \Delta_2), \Delta_3)] \t A}
\dip

&
\AC{\Gamma[(\Delta_1, \Delta_2), \Delta_3)] \t A}
\rl{,a2}
\UC{\Gamma[\Delta_1, (\Delta_2, \Delta_3)] \t A}
\dip
\\
& \\

\AC{\Gamma[\Delta_1; (\Delta_2; \Delta_3)] \t A}
\rl{;a1}
\UC{\Gamma[(\Delta_1; \Delta_2); \Delta_3)] \t A}
\dip

&
\AC{\Gamma[(\Delta_1; \Delta_2); \Delta_3)] \t A}
\rl{;a2}
\UC{\Gamma[\Delta_1; (\Delta_2; \Delta_3)] \t A}
\dip
\\
& \\

\AC{\Gamma[\Delta_1, \Delta_2] \t A}
\rl{,com}
\UC{\Gamma[\Delta_2, \Delta_1] \t A}
\dip

& 

\AC{\Gamma[\Delta_1; \Delta_2] \t A}
\rl{ent}
\UC{\Gamma[\Delta_1, \Delta_2] \t A}
\dip\\
&\\
\end{tabular}

\emph{Non-commutative connectives}

\begin{tabular}[]{cc}
& \\
\AxiomC{$\Gamma[A; B] \vdash C$} \RightLabel{$\odot$L} \UnaryInfC{$\Gamma[A \odot B] \vdash
C$} \DisplayProof  &  \AxiomC{$\Gamma \vdash A $} \AxiomC{$\Gamma' \vdash B$}
 \RightLabel{$\odot$R} \BinaryInfC{$\Gamma; \Gamma' \vdash A\odot B$} \DisplayProof      \\
& \\
\AxiomC{$\Gamma \vdash A$} \AxiomC{$\Delta[B] \vdash C$} \RightLabel{$\setminus$ L}
\BinaryInfC{$\Delta[\Gamma; A\setminus B] \vdash C$} \DisplayProof  &
\AxiomC{$A; \Gamma \vdash B$} \RightLabel{$\setminus$R}
\UnaryInfC{$\Gamma \vdash A \setminus B$}   \DisplayProof \\
& \\
\AxiomC{$\Gamma \vdash A$} \AxiomC{$\Delta[B] \vdash C$} \RightLabel{$/$ L}
\BinaryInfC{$\Delta[ B / A; \Gamma]  \vdash C$} \DisplayProof  &
\AxiomC{$\Gamma; A \vdash B$} \RightLabel{$/$R}
\UnaryInfC{$\Gamma \vdash A / B$}   \DisplayProof \\
\end{tabular}
\end{center}
\caption{\PCMILL: extending the sequent calculus \MILL}
\label{table:PCMILL}
\end{table}

\medskip\noindent\textbf{Semantics and completeness.}~
In order to define a class of modal Kripke resource models for \PCMILL, we
extend the models we have considered in
Section~\ref{sec:modal-krm}. We add to a modal Kripke resource model
${(M, e, \circ, \geq, N, V)}$ an associative, non-commutative
operation $\bullet$ such that $e$ is neutral also for $\bullet$. Thus,
a Kripke resource model is now specified by $\mathcal{M} = {(M, e,
  \circ, \bullet, \geq, N, V)}$.  Bifunctoriality is assumed also for
$\bullet$: if $m \geq n$, and $m' \geq n'$, then $m \bullet m' \geq n
\bullet n'$.  Moreover, the entropy principle is captured in the
models by means of the following constraint: for all $x$, $y$, $x
\circ y \geq x \bullet y$. If $\mathcal{M}$ satisfies all these
conditions, we call it a \emph{partially commutative modal Kripke resource
  model}.

The new truth conditions are the following:
\begin{description}
\item $m \models_\mathcal{M} A \odot B$ iff there exist $m_{1}$ and $m_{2}$ such that $m \geq m_{1} \bullet m_{2}$ and $m_{1} \models_\mathcal{M} A$ and  
$m_{2} \models_\mathcal{M} B$.
\item $m \models_\mathcal{M} A \setminus B$ iff for all $n \in M$, if $n \models_\mathcal{M} A$, then $n \bullet m \models_\mathcal{M} B$.  
\item $m \models_\mathcal{M} B / A$ iff for all $n \in M$, if $n \models_\mathcal{M} A$, then $m \bullet n \models_\mathcal{M} B$.
\end{description}
Note that, if $m \models A \otimes B$, then by $m_{1} \circ m_{2} \geq m_{1} \bullet m_{2}$ and heredity, we have that $m \models A \odot B$.

\medskip
We shall prove soundness and completeness of \PCMILL. We start by
discussing the partially commutative version of \MILL. Soundness of
\PCMILL wrt.\ the semantics above is just an extension of the induction
for the soundness of \MILL with the new rules for non-commutative
connective. For completeness, we need to extend the construction of
the canonical model in order to account for the non-commutative
structure.

As before, a context can be associated to a unique formula by means of
a recursive operation, here $.^+$. We adopt the following
definition:
\begin{align*}
()^+  &=  \textbf{1}\\
(A)^+  &=  A\\
(\Gamma,\Delta)^+  &=  (\Gamma^+ \otimes \Delta^+)\\
  (\Gamma;\Delta)^+  &=  (\Gamma^+ \odot\Delta^+)\\
\end{align*}
Let $\mathcal{M}^c_\bullet = (M^c,  e^c, \circ^c, \bullet^c, \geq^c, N^c, V^c)$ such that:
$M^c = \{ \Gamma \mid \Gamma \text{ is a partially}$\linebreak $\text{commutative context} \}$;
$e^c = ()$; 
$\Gamma \circ^c \Delta = (\Gamma, \Delta)$; 
$\Gamma \bullet^c \Delta = (\Gamma; \Delta)$;
$\Gamma \geq^c \Delta$ iff $\Gamma \vdash \Delta^+$;
$\Gamma \in V^c(p)$ iff $\Gamma \vdash p$;
$N^c(\Gamma) = \{ \mid \A \mid^c \mid \Gamma \vdash \Box\A\}$.


We will show that $\mathcal{M}^c_\bullet$ is actually a partially
commutative modal Kripke resource model. It suffices to adapt and extend
the proof of Lemma~\ref{lem:mc}. 
Important bits are:
\begin{enumerate}
\item $e^c$ is neutral for $\bullet^c$;
\item associativity of $\bullet^c$;
\item for all $\Gamma, \Delta \in M^c$: $\Gamma \circ^c \Delta \geq^c
  \Gamma \bullet^c \Delta$;
\item 
$\Gamma \geq^c \Delta$ and $\Gamma' \geq^c \Delta'$ then $\Gamma
\bullet^c \Gamma' \geq^c \Delta \bullet^c \Delta'$;
\item 
If $X \in N^c(\Gamma)$ and $\Gamma \geq^c \Delta$ then $X \in
N^c(\Delta)$.
\end{enumerate}
We sketch the arguments here. Item~1 and item~2 follow from the definition of partially commutative contexts. We look at the case of entropy (item~3) with more attention. By
 repeated use of (ax), $\otimes$R, and $\odot$R, we can show
$(\Gamma ; \Delta) \vdash (\Gamma ; \Delta)^+$. By (ent), we obtain
$(\Gamma , \Delta) \vdash (\Gamma ; \Delta)^+$, and apply the
definition of $\geq^c$ to have $(\Gamma , \Delta) \geq^c (\Gamma
; \Delta)^+$. By definition of $\circ^c$ we have $(\Gamma \circ^c
\Delta) \geq^c (\Gamma ; \Delta)^+$. Call $In$ the latter
inequality. Now, by (ax), we have $(\Gamma;\Delta)^+ \vdash (\Gamma;\Delta)^+$, which by definition of $\geq^c$ is equivalent to
$(\Gamma;\Delta)^+ \geq^c (\Gamma;\Delta)$. By definition of
$\bullet^c$, we have $(\Gamma ; \Delta)^+ \geq^c (\Gamma \bullet^c
\Delta)$. Together with $In$, we have 
$(\Gamma \circ^c
\Delta) \geq^c (\Gamma \bullet^c
\Delta)$. 

To start with item~4, assume $\Gamma \geq^c \Delta$ and $\Gamma'
\geq^c \Delta'$, that is by definition of $\geq^c$, $\Gamma \vdash
\Delta^+$ and $\Gamma' \vdash \Delta'^+$. By $\odot$R, we have
$(\Gamma ; \Gamma') \vdash \Delta^+ \odot \Delta'^+$. By the
definition of~$.^+$, it means that $(\Gamma ; \Gamma') \vdash (\Delta
; \Delta')^+$. Again by definition of $\geq^c$, we have $(\Gamma ;
\Gamma') \geq^c (\Delta ; \Delta')$. Finally by definition of
$\bullet^c$ we conclude that $(\Gamma \bullet^c \Gamma') \geq^c
(\Delta \bullet^c \Delta')$.

Item~5 is almost identical to the same case in the proof of Lemma~\ref{lem:mc}, but we explicitly adapt it here.
The hypotheses are
equivalent to $\Gamma \vdash \Box A$ for some $A$ such that $\mid A
\mid^c = X$, and $\Delta \vdash \Gamma^+$. By repeatedly applying
$\otimes$L and $\odot$L to obtain $\Gamma^+ \vdash \Box A$ and by using $cut$,
we infer that $\Delta \vdash \Box A$. Which is equivalent to the
statement that $X \in N^c(\Delta)$.


The truth lemma can be checked by routine induction. Thus we can conclude.
\begin{theorem} \PCMILL is sound and complete wrt.~the class of  partially commutative modal Kripke resource models. 
\end{theorem}

\section{Resource-sensitive ``bringing-it-about''}
\label{sec:linear-biat}

We present the (non-normal modal) logic of agency of
\emph{bringing-it-about}~\cite{Elgesem97agency,Governatori05Elgesem},
and propose two versions of it in Linear Logic coined \RSBIAT (for
Resource-Sensitive ``bringing-it-about'') and \SRSBIAT (for Resource-Sensitive
``bringing-it-about'' with sequences of actions). \RSBIAT and \SRSBIAT
are respectively extensions of \MILL and of \PCMILL. In Section~\ref{sec:application}, we will illustrate the logic by representing a few actions of
agents, functions of artefacts, and their interactions.

We specialise the minimal modality studied in the previous sections to
a modality agency. In fact, for each agent $a$ in a set $\mathcal{A}$,
we define a modality $\biat_{a}$, and $\biat_{a}A$ specifies that
agent $a \in \mathcal{A}$ brings about $A$. As previously, to
interpret them in a modal Kripke resource frame, we take one
neighbourhood function $N_a$ for each agent $a$ that obeys
Condition~(\ref{princ:heredityN}) in Definition~\ref{def:mkrm}.  We
have $m \models \biat_a A$ iff $|| A || \in N_a(m)$.

\subsection{Bringing-it-about in classical logic}
The four following principles typically constitute the core of logics
of agency~\cite{Prn77actionsocial,Elgesem97agency,belnap01facing}:

\begin{enumerate}\label{enum:principlesofagency}
\item If something is brought about, then this something holds.
\item It is not possible to bring about a tautology.
\item If an agent brings about two things concomitantly then the agent also brings about the conjunction of these two things.
\item If two statements are equivalent, then bringing about one is equivalent to bringing about the other.
\end{enumerate}
Briefly, we explain how these principles are captured in classical
logic. Item~1 is a principle of success. It corresponds to the axiom
T: $\biat_i A \rightarrow A$. Item~2 has been open to some debate,
although Chellas is essentially the only
antagonist. (See~\cite{chellas69imperatives} and~\cite{chellas92}.)
It corresponds to the axiom $\lnot \biat_i \top$ (notaut). Item~3
corresponds to the axiom: $\biat_i A \wedge \biat_i B \rightarrow
\biat_i (A \wedge B)$. That is, if $i$ is doing $A$ while also doing
$B$, then we can deduce that $i$ is doing $A\land B$. The other way
round needs not be true. Item~4 confers to the concept of bringing
about the quality of being a modality, effectively obeying the rule of
equivalents: if $\vdash A \leftrightarrow B$ then $\vdash \biat_i A
\leftrightarrow \biat_i B$.
\medskip

\subsection{Resource-sensitive BIAT}
We now detail the logic of \RSBIAT. We capture the four
principles, adapted to the resource-sensitive framework, by means of
rules in the sequent calculus, cf.~Table~\ref{table:biatsequent}

The principle of item~1 is captured by $\biat_{a}$(refl) that entails
the linear version of T: $\biat_{a} A \implies A$.  In our
interpretation, it means that if an agent brings about $A$, then $A$
affects the environment.

Because of the difference between the unities in Linear Logic and in
classical logic, the principle of item~2 requires some attention. In
classical logic all tautologies are provably equivalent to the unity
$\top$. Say $A$ is theorem ($\vdash A$), we have $\vdash A \leftrightarrow
\top$. Hence, from the rule of equivalents, and the axiom $\vdash
\lnot \biat_a \top$ that indicates than no agent brings about the
tautological constant, one can deduce $\vdash \lnot \biat_a A$
whenever the formula $A$ is a theorem. In Linear Logic, the unity
$\mathbf{1}$ is \emph{not} provably equivalent to all theorems. Thus,
the axiom of BIAT must be changed into an inference rule ($\sim$ nec)
in \RSBIAT: if $\vdash A$, then $\biat_a A\vdash \bot$. It is
effectively a sort of ``anti-necessitation rule''. So, if a formula is
a theorem, if an agent brings it about, then the
contradiction is entailed. This amounts to negating $\biat_a A$, according to
intuitionistic negation, for every tautological formula $A$.

The principle of BIAT for combining actions (item~3 in the list) is
open to two interpretations here: a multiplicative one and an additive one.
The additive combination means that if there is a choice for
agent $a$ between bringing about $A$ and bringing about $B$, then
agent $a$ can bring about a choice between $A$ and $B$. $\biat_a\otimes$ means that if an agent $a$ brings
about action $A$ and brings about action $B$ then $a$ brings about
both actions $A \otimes B$. Moreover, in order to bring
about $A \otimes B$, the sum of the resources for $A$ and the resources for $B$ is required. 

Finally, the logics of the minimal modality already satisfy the
rule of equivalents for $\biat_a$: from $A \t B$ and $B \t A $ we
infer $\biat_{a}A \t \biat_{a}B$. This is inherited by \RSBIAT,
and it is all that is needed to capture the principle of item~4.

We enrich \HMILL, the Hilbert system for \MILL, as follows and we label this
system \HRSBIAT.  We add the following axioms.
\begin{table}[h]
\begin{center}
\framebox{\begin{tabular}{l}
all axioms of \HMILL\\
$\biat_a A \implies A$\\
$\biat_a A \otimes \biat_a B \implies \biat_a (A \otimes B)$\\
$\biat_a A \with \biat_a B \implies \biat_a (A \with B)$
\end{tabular}}\end{center}
\caption{Axiom schemata in \HRSBIAT}
\end{table}

The definition of deduction in \HRSBIAT extends the definition of deduction in \HMILL (Definition~\ref{def:derivation-hmill}) with the following possible rule to consider for the inductive steps.
$$
\AC{\stackrel{\mathcal{D}}{\vdash_\H A}} \rl{($\sim$ nec)}
\UC{\vdash_\H \biat_a A \implies \bot}
\dip
$$

On the side of the semantics, we propose the following conditions on
modal Kripke resource frames $(M, e, \circ, \geq, \{N_a\}, V)$.
The rule ($\sim$ nec) requires: 
\begin{equation}\label{eq:cond-no1}
\text{if}\; (X \in N_a(w)) \text{ and } (e \in X) \text{ then } (w \in V(\bot))
\end{equation}
The rule ($\biat_a$(refl)) requires: 
\begin{equation}\label{eq:reflexivity}
\text{if }\; X \in N_a(w)\;\text{then}\; w \in X 
\end{equation}
The condition corresponding to the multiplicative version of action combination ($\biat_{a}\otimes$) is  the following, where $X \circ Y = \{x \circ y \mid x \in X\; \text{and}\; y \in Y\}$, and\label{page:uparrow} $X^\uparrow = \{ y  \mid y \geq x\; \text{and}\; x \in X\}$:
\begin{equation}\label{cond:tensor}
\text{if}\; X \in N_{a}(x)\; \text{and}\;  Y \in N_{a}(y)\; \text{, then}\;  (X \circ Y)^\uparrow \in N_{a}(x \circ y)
\end{equation}
The condition on the frames corresponding to the additive version is the following:
\begin{equation}\label{cond:with}
\text{if}\; X \in N_{a}(x)\; \text{and}\; Y \in N_{a} (x),\; \text{then}\; X \cap Y \in N_{a}(x) 
\end{equation}

Next, we introduce a sequent calculus for \RSBIAT. The rules of \RSBIAT for $\biat \with$ show that in order to bring about the choice between $A \with B$ is enough to use the resources for one of the two. On the contrary, in order to bring
about $A \otimes B$, the sum of the resources for $A$ and the resources for $B$ is required. 

\begin{table}[h]
\begin{center}

\begin{tabular}[]{cc}

\AC{\vdash A}  \rl{$\sim$nec} \UC{\biat_a A \vdash \bot} \DisplayProof &
\AC{\Gamma \vdash \biat_a A} \AC{\Gamma \vdash \biat_a B} \rl{$\biat_a\with$} \BinaryInfC{$\Gamma \vdash \biat_{a}(A \with B)$} 
\DisplayProof\\
& \\

\AC{\Gamma, A \vdash B} \rl{$\biat_{a}$(refl)} \UC{\Gamma, \biat_{a} A \vdash B} \dip & 
\AC{\Gamma \vdash \biat_a A} \AC{\Delta \vdash \biat_a B} \rl{$\biat_{a}\otimes$} \BC{\Gamma, \Delta \vdash \biat_{a} (A \otimes B)} \dip
\end{tabular}
\end{center}
\caption{\RSBIAT (extends \MILL)\label{table:biatsequent}}
\end{table}

We can prove that \HRSBIAT\ and the sequent calculus for \RSBIAT are equivalent.
\begin{prop}
It holds that $\Gamma \t_\HRSBIAT A$ iff $\Gamma \t A$ is derivable in \RSBIAT.
\end{prop}

 \begin{proof}(Sketch) The proof is again an induction on the length of derivations. For example, in one direction, axiom $\biat_a A \otimes \biat_a B \implies \biat_a (A \otimes B)$ is derivable in the sequent calculus by simply applying $\otimes L$ and $\implies R$ to  $\biat_a A, \biat_a B \t \biat_a (A \otimes B)$ which has been obtained from axioms by means of $\biat_a\otimes$.
 \end{proof}

We can now prove soundness and completeness of \RSBIAT.
\begin{theorem} \RSBIAT is sound and complete wrt.~the class of modal Kripke frames that satisfy (\ref{eq:cond-no1}), (\ref{eq:reflexivity}), (\ref{cond:tensor}), and (\ref{cond:with}).
\end{theorem}

\begin{proof} (Sketch) We just show two correspondences. 

\emph{Condition~(\ref{eq:cond-no1}) and rule ($\sim$nec).} ($\sim$nec) is sound. 
Assume 
that for every model, $e \models A$. We need to show that $e \models \biat_a A \implies
\bot$. That is, for every $x$, if $x \models \biat_a A$, then $x$
models $\bot$. If $x \models \biat_a A$, then by definition, $||A||
\in N_a(x)$. Since $A$ is a theorem, $e \in ||A||$, thus by Condition
\ref{eq:cond-no1}, $x \in V(\bot)$, so $x\models \bot$.
For completeness, it suffices to adapt our canonical model
construction. Build the canonical model for \RSBIAT as in
Def.~\ref{def:mc} (we have now more valid sequents). Now suppose
(1)~$X \in N^c_a(\Gamma)$, and (2)~$e^c \in X$. By definition of $N^c_a$
and of $\mid . \mid^c$, there is $A$, s.t.\ $\mid A\mid^c = X$, (1) $\Gamma \vdash
\biat_{a}A$ and (2) $\vdash A$. From (2), and ($\sim$nec): $\biat_{a}A
\vdash \bot$. From (1), and previous, we obtain $\Gamma \vdash \bot$
using (cut). By definition of $V^c$, $\Gamma \in V^c(\bot)$.

\par

\emph{Condition~(\ref{cond:tensor}) and rule ($\biat_{a}\otimes$).}
($\biat_{a}\otimes$) is sound. Assume $e \models \Gamma^* \implies
\biat_a A$ and $e \models \Delta^* \implies \biat_a B$. Then, for all
$x$ that make $\Gamma$ true, $||A|| \in N_a(x)$ and for all $y$ that
make $\Delta$ true, $||B|| \in N_a(y)$. By (\ref{cond:tensor}), $||A||
\circ ||B|| \in N_a(x \circ y)$, so for any $x \circ y$ that make
$(\Gamma, \Delta)^*$ true, $x \circ y \models \biat_a (A \otimes
B)$. For completeness, suppose $X \in N^c_a(\Gamma)$ and $Y \in
N^c_a(\Delta)$. By definition of $N^c_a$ and of $\mid . \mid^c$, there
is $A$ and $B$, with $\mid A\mid^c = X$, and $\mid B\mid^c = Y$, s.t.~$\Gamma \vdash \biat_{a}A$, and $\Delta \vdash \biat_{a}B$. By ($\biat_{a}\otimes$), we obtain $\Gamma, \Delta \vdash \biat_a(A \otimes B)$ and thus
$\mid A \otimes B \mid^c \in N^c_a(\Gamma \sqcup \Delta)$ by definition of $\mid . \mid^c$. The definition of $\circ^c$ gives us $\mid A \otimes B \mid^c \in N^c_a(\Gamma \circ^c \Delta)$. By the Truth Lemma, we have that $||A \otimes B||^{\mathcal{M}^c} \in N^c_a(\Gamma \circ^c \Delta)$. Thus $(X \circ^c Y)^\uparrow \in N^c_a(\Gamma \circ^c \Delta)$.

\end{proof}

Therefore, also our extensions of Hilbert systems and sequent calculus are sound and complete wrt.~the modal Kripke resource frames restricted to the relevant conditions given in this section. 

\begin{corollary} \HRSBIAT\ is sound and complete wrt.~the class of modal Kripke  frames that satisfy (\ref{eq:cond-no1}), (\ref{eq:reflexivity}), and (\ref{cond:tensor}).
\end{corollary}

Moreover, \RSBIAT enjoys cut elimination.
\begin{theorem} \label{thm:cut-elim-linearBIAT} Cut elimination holds for \RSBIAT.
\end{theorem}

\begin{proof}(Sketch) We extend the proof of Theorem (\ref{thm:cut-elim}) by presenting a number of new cases for the cut formula being principal in both premises of the cut rule. The other cases can be treated similarly. 
The cut formula has been introduced by $\biat_a$(re) and $\sim$nec

$$
\begin{tabular}{ccc}
\AC{A \t B}
\AC{B \t A}
\rl{$\biat_a$(re)}
\BC{\biat_a A \t \biat_a B}
\AC{\t B}
\rl{$\sim$nec}
\UC{\biat_a B \t \bot}
\rl{cut}
\BC{\biat_a A \t \bot}
\dip

&

$\rightsquigarrow$

&

\AC{\t B}
\AC{B \t A}
\rl{cut}
\BC{\t A}
\rl{$\sim$nec}
\UC{\biat_a A \t \bot}
\dip
\\
\end{tabular}
$$

The cut formula has been introduced by $\biat_a\otimes$ and $\biat_a$(refl).

$$
\AC{\Gamma' \t A}
\AC{...}
\rl{}
\BC{\Gamma \t \biat_a A}
\AC{\Delta' \t B}
\AC{...}
\rl{}
\BC{\Delta \t \biat_a B}
\rl{}
\BC{\Gamma, \Delta \t \biat_a (A \otimes B)}
\AC{\Sigma', A \t C'}
\AC{\Sigma'',B \t C''}
\rl{}
\BC{\Sigma, A \otimes B \t C}
\rl{}
\UC{\Sigma, \biat_a (A \otimes B) \t C}
\rl{cut}
\BC{\Gamma, \Delta, \Sigma \t C}
\dip
$$

It can be reduced by pushing the cut upwards.

$$
\AC{\Gamma' \t A}
\AC{\Sigma', A \t C'}
\rl{cut}
\BC{\Gamma',\Sigma' \t C'}
\AC{\Delta' \t B}
\AC{\Sigma'',B \t C''}
\rl{cut}
\BC{\Delta', \Sigma'' \t C''}
\rl{...}
\BC{...}
\rl{...}
\UC{\Gamma, \Delta, \Sigma \t C}
\dip
$$
\end{proof}

Once again, it is easy to see that cut elimination entails the subformula property for \RSBIAT. Using the same arguments as for Theorem~\ref{thm:PSPACE}, it is clear that we can decide polynomial space whether a sequent is valid in \RSBIAT.
\begin{theorem} Proof search complexity for \RSBIAT is in PSPACE.
\end{theorem}

\subsection{\RSBIAT with sequences of actions}
\label{sec:lbiat-sa}

So far, we have discussed how to control weakening and contraction in order to provide a resource sensitive account of agency. There is one important structural rule that we did not discuss, namely the exchange rule. In this section, we extend \RSBIAT by introducing the \emph{non-commutative} multiplicative conjunction $\odot$, and its two associated order-sensitive implications.

The significance of this move goes beyond the technical aspect, which
is rather straightforward at that point.  Indeed, a recurring point of
contention against the logics of bringing-it-about is the absence of a
basic notion of time.  Non-commutative composition of formulas will
provide to us an immediate and natural way of talking, if not about
time proper, at least about sequences of actions. Reading $A \odot B$
as ``first $A$ then $B$'', we will also read $(\biat_a A) \odot
(\biat_b B)$ as ``first $a$ brings about $A$ then $b$ brings about
$B$''.

\medskip

We obtain \SRSBIAT by adding the non-commutative versions of  $\biat_a\otimes$ and rephrasing the commutative rules by means of the context notation. The rule $\biat_a$(refl) can now operate both in commutative and non-commutative contexts. 

\begin{table}[h]
\begin{center}
\begin{tabular}[]{cc}
\AC{\Gamma[A] \vdash B} \rl{$\biat_{a}$(refl)} \UC{\Gamma[\biat_{a} A] \vdash B} \dip & 
\AC{\Gamma \vdash \biat_a A} \AC{\Delta \vdash \biat_a B} \rl{$\biat_{a}\odot$} \BC{\Gamma; \Delta \vdash \biat_{a} (A \odot B)} \dip\\
\end{tabular}
\end{center}
\caption{Resource ``bringing-it-about'' with sequences of actions (extends PCL and \HRSBIAT)}
\end{table}

Note that the presentation of $\sim$nec, $\biat_a\otimes$, and $\biat_a\with$ is not affected 
by the generalisation to partially commutative. 
$\biat_{a}$(refl) states that we can introduce the modality also in non-commutative contexts. $\biat_{a}\odot$ adapts the principle of composition of actions (cf.~item~3, p.~\pageref{enum:principlesofagency}) to sequences; it states that we can compose two ordered actions into one sequence of actions. 

In order to offer a semantics to our extension of \RSBIAT with sequences of actions, we need to add a condition on the neighbourhood functions that deals with the non-commutative operator. Specifically, we need the following condition:
\begin{equation}\label{cond:next}
\text{if}\; X \in N_{a}(x)\; \text{and}\;  Y \in N_{a}(y)\; \text{, then}\;  (X \bullet Y)^\uparrow \in N_{a}(x \bullet y)
\end{equation}

We obtain naturally the semantic determination.
\begin{theorem} \SRSBIAT is sound and complete wrt.~the class of partially commutative modal Kripke resource models that satisfy 
 (\ref{eq:cond-no1}), (\ref{eq:reflexivity}), (\ref{cond:tensor}),
  (\ref{cond:with}), and (\ref{cond:next}).
\end{theorem}

\section{Application: manipulation of artefacts}
\label{sec:application}

\subsection{Artefacts}

Our application lies in the
reasoning about artefact's function and tool use. By endorsing what we may call an \emph{agential} stance, we view artefacts as
special kind of agents. They are characterised by the fact that they
are designed by some other agent in order to achieve a purpose in a
particular environment.  An important aspect of the modelling of
artefacts is their interaction with the environment and with the
agents that \emph{use} the artefact to achieve a specific
goal~\cite{garbacz04,borgoetal09handbook,HoukesVermaas2010,Kroes2012}.
Briefly, we can view an artefact as an object that in presence of a
number of preconditions $c_{1}, \dots, c_{n}$ produces the outcome
$o$. In this work, we want to represent the function of artefacts by
means of logical formulas and to view the correct behaviour of an
artefact by means of a form of reasoning. When reasoning about
artefacts and their outcomes, we need to be careful in making all the
conditions of use of the artefact explicit, otherwise we end up facing
the following unintuitive cases. Imagine we represent the behaviour of
a screwdriver as a formula of classical logic that states that if
there is a screw $S$, then we can tighten it $T$.  We simply describe
the behaviour of the artefact as a material implication $S \rightarrow
T$.  In classical logic, we can infer that by means of a single
screwdriver we can tighten two screws: $S, S, S \rightarrow T \vdash T
\wedge T$. Worse, we do not even need to have two screws to begin
with: $S, S \rightarrow T \vdash T \wedge T$. Thus, without specifying
all the relevant constraints on the environment (e.g., that a
screwdriver can handle one screw at the time) we end up with
unintuitive results. Another possible drawback of classical logic is
that it is commutative, the order of formulas does not matter.  For
example, if we describe the process of hammering a nail by means of
the implication if I place a nail $N$ and I provide the right force
$F$, then I can drive a nail in ($D$), that is $N \wedge F \rightarrow
D$, that would entail also that one can put a force before placing the
nail.

Moreover, we need to specify the relationship between the artefact and
the agents: for example, there are artefacts that can be used by one
agent at the time.  Since a crucial point in modelling artefacts is
their interaction with the environment, either we carefully list all
the relevant conditions, or we need to change the logical framework
that we use to represent the artefact's behaviour.  In this paper, we
propose to pursue this second strategy.  Our motivation is that,
instead of specifying for each artefact the precondition of its
application (e.g., that there is only one screw that a screw driver is
supposed to operate on), the logical language that encodes the
behaviour of the artefact already takes care of preventing unintuitive
outcomes. Thus, the formulas of Linear Logic shall represent actions
of agents and functions of artefacts, and the non-normal modality
shall specify which agent or artefact brings about which process.

\subsection{Functions}

The concept of a \emph{function} of an artefact aims to capture the
description of the behaviour of an artefact in an environment with
respect to its goals: artefacts are not living things but have a
purpose, attributed by a designer or a
user~\cite{borgoetal09handbook,Kroes2012}.  We model a function of an
artefact by means of a formula $A$ in \RSBIAT or \SRSBIAT.  If $A$ is a function
of an artefact $t$, then one can represent $t$'s behaviour as $\biat_t
A$ ($t$ brings about $A$) in a conceptually consistent manner, namely
an artefact brings about its function $A$.

With Linear Logic, we are equipped with a formalism to represent and
reason about processes and resources. In classical and intuitionistic
logic, if one has $A$ and $A$ implies $B$, then one has $B$, but $A$
still holds. This is fine for mathematical reasoning but often fails
to be acceptable in the real world where implication is
\emph{causal}. Girard remarks that ``[a] causal implication cannot be
iterated since the conditions are modified after its use; this process
of modification of the premises (conditions) is known in physics as
\emph{reaction}.''~\cite[p.~72]{girard89}
That is, Linear Logic allows for modelling how the function of an artefact can be actually realised in a certain environment.
At an abstract level, an artefact can be seen as an agent  $t$. It takes resource-sensitive
actions by reacting to the environment.  For any artefact $t$ with function $A$, $\biat_t A $, we say that $t$ accomplishes a certain goal $O$ in the environment $\Gamma$ if and only if the sequent 
$\Gamma[\biat_t A] \t O$ is provable. The context $\Gamma$ describes a number of preconditions that specify the environment resources as well as the actions of the agents that are interacting with the artefact. 
\dnote{The proof  of $\Gamma[\biat_t A] \t O$ exhibits the execution of the function of $t$ in the environment $\Gamma$: when $t$ is an artefact, and $\Gamma[\biat_t A] \t O$ is provable, then the occurrence of $A$ in the proof of $\Gamma[\biat_t A] \t O$ is the concrete instantiation of the function of $t$ in $\Gamma$.  Since we are using intuitionistic versions of Linear Logic, every proof of a sequent is a process that behaves like a function in the mathematical sense.\footnote{This suggestion can be made mathematically precise. It is possible to associate terms in the (linear) $\lambda$-calculus to proofs in intuitionistic Linear Logic \cite{Benton1993term}.} Hence, the rules of sequent calculus provide instructions to compose basic functions of artefacts to obtain complex functions and to model how the composition interacts with the environment. We will see examples of complex functions in the next paragraphs.} For that reason, our view of artefacts simply generalises to a number of artefacts that interact in an environment as follows: 


$$\Gamma[\biat_1 A_1, \dots,\biat_m A_m] \vdash O$$ 

Again, if the above sequent is provable, then the combination of artefacts $ \biat_1 A_1$, ..., $\biat_m A_m $ can achieve the goal $O$ in $\Gamma$ by executing their functions $A_1, \dots, A_m$ in $\Gamma$. 

Defining the function of an artefact as a formula demands some care because in this way functions do not have a
unique formulation. The functions $(A \otimes B) \implies C$, and $A
\implies (B \implies C)$ are provably equivalent. However, the
rule~$\biat_a$(re) ensures that bringing about a function
is provably equivalent to bringing about any of its
equivalent forms.\\ 
By means of sequent calculus provability, we can view the problem of using artefacts in an environment to achieve a goal as a decision problem that is related to the AI problem of planning \cite{KanovichVauzeilles2001}. 
Note that the complexity of deciding whether a goal is achievable depends only on the fragment of the logic that we use to model the formulas in the sequent. 
In the next paragraphs, we shall instantiate the descriptive features of our calculus by means of a number of toy examples. 



\subsection{Simple examples of functions}
\label{sec:screwdriver}

Take a very simple example. We can represent the function of a
screwdriver $s$ as an implication that states that if there is a screw
(formula $S$) and some agent brings about the right rotational force
($F$), then the screw gets tighten ($T$). The formula corresponding to
the function of the screwdriver is $S \otimes F \implies T$. The
formula that captures the screwdriver as an agent of the system is
$\biat_s (S \otimes F \implies T)$.

Suppose the environment provides $S$ and an agent $i$ is providing the
right force $\biat_{i} F$. We can show by means of the following proof
in \RSBIAT that the goal $T$ can be achieved.
\[
\AC{S \t S} \rl{}
\AC{F \t F}\rl{$\biat_{i}$(refl)}
\UC{\biat_{i}F \t F}\rl{$\otimes$R}
\BC{S, \biat_{i}F \t S \otimes F}
\AC{T \t T}
\rl{$\implies$L}
\BC{S, \biat_{i}F, S \otimes F \implies T \t T}
\rl{$\biat_{s}$(refl)}
\UC{S, \biat_{i}F, \biat_s (S \otimes F \implies T) \t T}
\dip   
\]
Our calculus is resource  sensitive, thus, as expected, we cannot infer for example that two agents can use the same screwdriver at the same time to tighten two screws: $$S, S,  \biat_{i} F, \biat_{j} F,  \biat_s (S \otimes F \implies T) \not\t T \otimes T$$
That would in fact require two screwdrivers and thus an extra $\biat_s (S \otimes F \implies T)$ at the left of the sequent. 
Often, to be effective for some goal $B$, an artefact's function transforming a resource $A$ into a resource $B$ should not be
realised before the resource $A$ is available. 

In the case of our example, the description of a screwdriver should exclude that the screw can be tighten \emph{before} a loose screw and a rotational force, in this order, are provided. Thus, we may reconsider our screwdriver as a ``non-commutative'' screwdriver $s\bullet$ and write its function as $S \odot F \setminus T$.  The screwdriver is now defined as $\biat_{s\bullet}(S \odot F \setminus T)$.  

\[
\AC{S \t S} \rl{}
\AC{F \t F}\rl{$\biat_{i}$(refl)}
\UC{\biat_{i}F \t F}\rl{$\odot$R}
\BC{S; \biat_{i}F \t S \odot F}
\AC{T \t T}
\rl{$\setminus$L}
\BC{(S; \biat_{i}F); S \odot F \setminus T \t T}
\rl{$\biat_{s\bullet}$(refl)}
\UC{S; \biat_{i}F; \biat_{s\bullet} (S \odot F \setminus T) \t T}
\dip   
\]

The meaning of entropy (ent) is the following. By means of (ent),  we can infer a fully commutative context:

$$S, \biat_{i}F, \biat_{s\bullet} (S \odot F \setminus T) \t T$$

That means that it is the description of the function of the artefact that takes care of of specifying how the resources have to be ordered. Of course, our screwdriver formula correctly excludes, for instance, that the force is applied after using the screwdriver: $S; \biat_{s\bullet} (S \odot F \setminus T); \biat_{i}F \not\t T$ Moreover, $\biat_{i}F; \biat_{s\bullet} (S \odot F \setminus T); S \not\t T$, and $\biat_{i}F; S; \biat_{s\bullet} (S \odot F \setminus T) \not\t T$.\\

The interaction of commutative and non-commutative operators is exemplified as follows. Suppose there are two screws $S$, $S$, two screwdrivers $s$ and $s'$ and two agents $a$, $b$. The goal of tightening two screws can be achieved by using the screwdrivers in whatever order, as the following proof shows. 

\footnotesize

\[
\AC{S \t S} \rl{}
\AC{F \t F}\rl{$\biat_{i}$(refl)}
\UC{\biat_{a}F \t F}\rl{$\odot$R}
\BC{S; \biat_{a}F \t S \odot F}
\AC{T \t T}
\rl{$\setminus$L}
\BC{(S; \biat_{a}F); S \odot F \setminus T \t T}
\rl{$\biat_{s}$(refl)}
\UC{S; \biat_{a}F; \biat_{s} (S \odot F \setminus T) \t T}
\AC{S \t S} \rl{}
\AC{F \t F}\rl{$\biat_{b}$(refl)}
\UC{\biat_{b}F \t F}\rl{$\odot$R}
\BC{S; \biat_{b}F \t S \odot F}
\AC{T \t T}
\rl{$\setminus$L}
\BC{(S; \biat_{b}F); S \odot F \setminus T \t T}
\rl{$\biat_{s'}$(refl)}
\UC{S; \biat_{b}F; \biat_{s'} (S \odot F \setminus T) \t T}
\rl{$\otimes$R}
\BC{[S; \biat_{a}F; \biat_{s} (S \odot F \setminus T)],[S; \biat_{b}F; \biat_{s'} (S \odot F \setminus T) ] \t T \otimes T}
\dip
\]

\normalsize

\subsection{Functions composition}

By extending the previous example, we can demonstrate how the output
of some artefact's function can naturally be fed into another function so as to construct a new complex artefact.

An \emph{electric screwdriver} has two components.  Firstly, the
\emph{power-pistol} creates some rotational force $F$ when the
button is pushed ($P$): $P \setminus F$. Secondly, what is
typically called the screwdriver \emph{bit} is for all intents and
purposes effectively a screwdriver as specified before: it tightens a
loose screw when a rotational force is applied. 
We define the electric screwdriver by means of  $\biat_e ((P \setminus F) \odot (S \odot F \setminus T))$

Now suppose the environment provides a loose screw $S$ and an agent
$i$ is pushing the button of the power pistol: $\biat_{i} P$. We can
show again that the goal $T$ of having a tighten screw can be achieved, by using the electric screwdriver.

\[
\rl{}
\AC{S \t S}
\AC{P \t P}
\rl{$\biat_{i}$(refl)}
\UC{\biat_i P \t P}
\AC{F \t F}
\rl{$\setminus$L}
\BC{ \biat_{i}P; P \setminus F \t F}
\rl{$\odot$R}
\BC{S; \biat_{i}P; P \setminus F \t S \odot F}
\AC{T \t T}
\rl{$\setminus$L}
\BC{S; \biat_{i}P; P \setminus F, S \odot F \setminus T   \t T}
\rl{$\odot$L}
\UC{S; \biat_{i} P; (P \setminus F) \odot (S \odot F \setminus T) \t T}
\rl{$\biat_e$(refl)}
\UC{S; \biat_{i} P; \biat_e ((P \setminus F) \odot (S \odot F \setminus T)) \t T}
\dip   
\]

\normalsize

\subsection{Complex interactions between agents and artefacts}

The function of an artefact may require to specify how agents use it. A number of aspects of the interaction between agents' actions and artefacts can be captured by means of the rules $\biat_i \otimes$, $\biat_i \odot$, and $\biat_i \with$. Recall that $\mathcal{A}$ is a set of agents. We write $\bigwith_{x \in \mathcal{A}} \biat_x A$ as a short hand for $ \biat_{i_1} A \with \dots \with \biat_{i_m} A$. The latter formula means that any agent can perform $A$, so for example
$ \biat_{i_1} A \with \dots \with \biat_{i_m} A \t \biat_{i_j} A$.\\

An artefact that is defined by  $\biat_t(\bigwith_{x \in \mathcal{A}} (\biat_x(A \otimes B) \implies O))$ requires the \emph{same} agent $x$ to perform both actions $A$ and $B$ in order to get $O$. For example, a one person rowboat that requires a single agent to operate on both oars ($R_1$) and ($R_2$), in whatever order, so to produce movement ($M$).
This is can be modelled by means of our $\biat_i \otimes$ rule. 

$$
\AC{\biat_i R_1 \t \biat_i R_1}
\AC{\biat_i R_2 \t \biat_i R_2}
\rl{$\biat_i \otimes$}
\BC{\biat_i R_1, \biat_i R_2 \t  \biat_i (R_1 \otimes R_2)}
\AC{M \t M}
\rl{$\implies$L}
\BC{\biat_i R_1, \biat_i R_2,  \biat_i ((R_1 \otimes R_2) \implies M) \t M}
\rl{$\with$L  enough times}
\UC{\biat_i R_1, \biat_i R_2, \bigwith_{x \in \mathcal{A}} \biat_i ((R_1 \otimes R_2) \implies M) \t M}
\rl{$\biat_t$(refl)}
\UC{\biat_i R_1, \biat_i R_2,  \biat_t (\bigwith_{x \in \mathcal{A}} \biat_i ((R_1 \otimes R_2) \implies M) ) \t M}
\dip
$$

On the other hand, by specifying the function by $\biat_t(\bigwith_{x, y \in \mathcal{A}, x \neq y} (\biat_x A \otimes \biat_y B) \implies O)$, we are forcing the agents who operate tool $t$ to be different (e.g., a crosscut saw). If an artefact's function does not determine whether the actions must be performed by the same agent, we can write $\biat_t(\bigwith_{x, y \in \mathcal{A}} (\biat_x A \otimes \biat_y B) \implies O)$.

In the non-commutative case, $\biat_t(\bigwith_{x \in \mathcal{A}} \biat_x(A \odot B) \implies O)$ forces the same agent to perform first $A$ and then $B$, whereas $\biat_t(\bigwith_{x, y \in \mathcal{A}, x \neq y} (\biat_x A \odot \biat_y B) \implies O)$ forces the agents to be different. For example, the function of a hanging ladder is described as follows: firstly, an agent holds the laden ($Ho$), then another agent climbs up ($Cl$) and reaches a certain position $R$: $\biat_t (\biat_a Ho \odot \biat_b Cl \setminus \biat_b R)$.
$$
\AC{\biat_a Ho \t \biat_a Ho}
\AC{\biat_b Cl \t \biat_b Cl}
\rl{$\odot$R}
\BC{\biat_a Ho; \biat_b Cl \t \biat_a Ho \odot \biat_b Cl}
\AC{\biat_b R \t \biat_b R}
\rl{$\setminus$L}
\BC{\biat_a Ho; \biat_b Cl; \biat_a Ho \odot \biat_b Cl \setminus \biat_b R\t \biat_b R}
\rl{$\biat_t$(refl)}
\UC{\biat_a Ho; \biat_b Cl; \biat_t(\biat_a Ho \odot \biat_b Cl \setminus \biat_b R)\t \biat_b R}
\dip
$$

By means of $\biat_i\with$, we can describe a function that requires an agent's choice. For example, a monkey wrench can tighten two sizes of nuts/bolts ($N_1$, $N_2$) provided  that an agent chooses the right measure ($M_1$, $M_2$): 
$\biat_t( (\biat_i (M_1 \with M_2) \implies \biat_i N_1) \with (\biat_i (M_1 \with N_2) \implies  \biat_i N_2))$. The following proof shows that if a single agent can choose the right measure for the nut (e.g., $M_1$), then the same agent can tighten the right type of nut (e.g., $\biat_i N_1$)

\footnotesize
$$
\AC{\biat_i M_1 \t \biat_i M_1}
\rl{$\with$L}
\UC{\biat_i M_1 \with \biat_i M_2 \t \biat_i M_1}
\AC{\biat_i M_2 \t \biat_i M_2}
\rl{$\with$R}
\UC{\biat_i M_1 \with \biat_i M_2  \t \biat_i M_2}
\rl{$\biat_i\with$}
\BC{\biat_i M_1 \with \biat_i M_2  \t \biat_i (M_1 \with M_2)}
\AC{\biat_i N_1 \t \biat_i N_1}
\rl{$\implies$L}
\BC{\biat_i M_1 \with \biat_i M_2 , \biat_i (M_1 \with M_2) \implies \biat_i N_1 \t \biat_i N_1}
\LeftLabel{$\with$L}
\UC{\biat_i M_1 \with \biat_i M_2 , (\biat_i (M_1 \with M_2) \implies \biat_i N_1) \with (\biat_i (M_1 \with M_2) \implies  \biat_i N_2)  \t \biat_i N_1}
\LeftLabel{$\biat_t$(refl)}
\UC{\biat_i M_1 \with \biat_i M_2 , \biat_t((\biat_i (M_1 \with M_2) \implies \biat_i N_1) \with (\biat_i (M_1 \with M_2) \implies  \biat_i N_2)) \t \biat_i N_1}
\dip
$$

\normalsize
In a similar way, we can represent functions of artefacts that require any number of actions and of agents to achieve a goal (of course, if we want to express that any subsets of $\mathcal{A}$ can operate the tool, then we need an exponentially long formula).




\subsection{Function warranty and reuse of artefacts}

\RSBIAT is resource-sensitive as the non-provable sequent in our screwdriver example
illustrated (Sec.~\ref{sec:screwdriver})
\[S, S,  \biat_{i} F, \biat_{j} F,  \biat_s (S \otimes F \implies T) \not\t T \otimes T\]
The screwdriver \emph{cannot be reused}, despite the fact that an
additional screw is available and an appropriate force is brought
about. This is perfectly fine as long as our interpretation of
resource consumption is \emph{concurrent}: all resources are consumed
at once. And indeed, one cannot tighten two screws at once with only
one screwdriver.

Abandoning a concurrent interpretation of resource consumption, we may
specialise the modality $\biat_a$ when $a$ is an artefactual agent in
such a way that the function of an artefact can be used at will. After
all, using a screwdriver once does not destroy the screwdriver. Its
function is still present after. It seems that we are after a property of
\emph{contraction} for our operator $\biat_s$.
\begin{center}
\AC{\Gamma, \biat_s A, \biat_s A \vdash B} \rl{c($\biat_s$)} \UnaryInfC{$\Gamma, \biat_s A \vdash B$} \DisplayProof
\end{center}
Now, if we adopt the rule, c($\biat_{s}$) we can easily see that
indeed $$S, S, \biat_{i} F, \biat_{j} F, \biat_{s} (S \otimes F \implies
T) \vdash T \otimes T$$ is provable.

There are several issues with this solution to `reuse' as a
duplication of assumptions. Some technical, some conceptual. The main
technical issue is that we lose a lot of control on the proof search,
as contraction is the main source of non-termination (of bottom-up
proof search). Another technical (or theoretical) issue is that trying
to give a natural condition on our frames that would be canonical for
contraction is out of question. The conceptual issue is the same as
the one posed by Girard in creating Linear Logic: duplication of
assumptions should not be automatic. Similarly, \emph{ad lib} reuse of
an artefact does not reflect a commonsensical experience. In general,
although they don't consume after the first use, tools will nonetheless
eventually become so worn out that they will not realise their
original function. \dnote{The point is that, in order to keep track of the relevant resources, the reuse of artefacts should not be arbitrary and should be allowed in a controlled manner instead.}

We can capitalise on the `additive' feature of Linear Logic language:
employing the `with' operator $\&$, we can specify a sort of warranty of
artefact functions. Denote $A^n = A \odot \dots \odot A$, for $n$ times. 
\newcommand\warranty[2]{\ensuremath{{#1}^{\leq #2}}}
We present the treatment by focusing on our example of screwdriver. 
A sequentially reusable screwdriver is defined as follows:

\[\warranty{(S\odot F \setminus T)}{n} = (S \odot F \setminus T) \with (S^2 \odot F^2 \setminus T^2) \with \dots \with (S^n \odot F^n \setminus T^n) 
\]

For example, with three screws, and three agents ($a$, $b$, and $c$) providing the appropriate force, then using a decently robust screwdriver, one
can obtain three tighten screws. We have:
\[S;S;S;\biat_a F; \biat_ b F; \biat_c F; \biat_{s\bullet}(\warranty{(S\odot F \setminus T)}{10000}) \vdash T \odot T \odot T
\]
where ${s\bullet}$ is our ``non-commutative'' screwdriver, now with a ten thousand-use warranty.

\footnotesize
$$
\AC{S \t S}
\AC{S \t S}
\rl{}
\BC{S; S \t S \odot S}
\AC{S \t S}
\rl{}
\BC{S; S; S \t S \odot S \odot S}
\AC{\biat_a F \t F}
\AC{\biat _b F\t F}
\rl{}
\BC{\biat_a F; \biat_b F \t \biat_a F \odot \biat_b F}
\AC{\biat_c F \t  F}
\rl{}
\BC{\biat_a F; \biat_b F; \biat_c F \t F \odot  F \odot F}
\rl{}
\BC{S; S; S; \biat_a F; \biat_b F; \biat_c F \t S \odot S \odot S \odot F \odot  F \odot F}
\AC{T^3 \t T^3}
\rl{}
\BC{S; S; S; \biat_a F; \biat_b F; \biat_c F; S^3 \odot F^3 \setminus T^3 \t T^3}
\rl{$\with$L}
\UC{S; S; S; \biat_a F; \biat_b F; \biat_c F; \warranty{(S\odot F \setminus T)}{10000} \t T^3}
\rl{$\biat_s$(re)}
\UC{S; S; S; \biat_a F; \biat_b F; \biat_c F; \biat_s(\warranty{(S\odot F \setminus T)}{10000}) \t T^3}
\dip
$$

\normalsize

Note that, the goal $T \otimes T \otimes T$ is not provable, and this reflects our view of reusability as a sequential operation. 



By pushing this analogy of warranty of artefact functions further on,
we can model a ``refurbishing'' function that augments the warranty of
the function $A$ of a tool $t$. For instance, consider the
refurbishing function which at the cost of consuming a resource $R$,
transforms a worn out (but not too much worn out!) screwdriver $t$
into a screwdriver $t$ with a function with extended warranty. It can
be written as:
\[ R \otimes \biat_s(\warranty{(S\odot F \setminus T)}{50}) \implies \biat_s(\warranty{(S\odot F \setminus T)}{7000})\]
and is for example a function that a bench grinder would have.

\section{Conclusion}
The semantics of all the languages studied in this paper are adequate
extensions of Urquhart's Kripke resource models for intuitionistic substructural 
logic. We first enriched the Kripke resource models with a
neighbourhood function to give a meaning to a minimal (non-normal)
modality. We obtained what we simply coin modal Kripke resource
models. We thus defined and studied a minimal non-normal modal
logic. The non-normal minimal modality $\Box$ is defined in the usual
way where $\Box A$ holds at a point of evaluation iff the extension of
$A$ is in the neighbourhood of the point of evaluation. With
Condition~\ref{princ:heredityN} on the models we enforced a sort of
heredity (or monotonicity) on the neighbourhood function, without which
the logic would not even validate modus ponens.

We introduced a Hilbert system and a sequent calculus. We showed that
both are sound and complete with respect to the class of modal Kripke
resource models. Moreover, we showed that the sequent calculus admits
cut elimination. We used this fact to establish that proof search can
be done in PSPACE. The soundness and completeness of the sequent
calculus, and the cut elimination show that the rule $\Box$(re)
\begin{center}
\AC{A \t B}
\AC{B \t A}\rl{}
\BC{\Box A \t \Box B}
\dip
\end{center}
which is atypical in a sequent calculus, yields the expected results
and presents no logical issue. Moreover, the new semantics and calculi
are general enough to extend the framework, as we did, to account for
partially commutative Linear Logic.

The significance of non-normal modal logics and their semantics in
modern developments in logics of agents has been emphasised before in
the literature.  In classical logic, logics of agency in particular,
have been widely studied and used in practical philosophy and in
multi-agent systems. Moving from classical logic to resource-sensitive
logics allows to lift the study of agents bringing about states of
affairs to the study of agents bringing about resources, or of
artefacts bringing about resource transformations.

We thus instantiated the minimal modal logics (commutative / and
partially commutative) with a resource-sensitive version of the logics
of bringing-it-about. Again, sound and complete Hilbert system and
sequent calculus are provided. Proof-search in the commutative version
is shown to be PSPACE-easy. We finally presented a number of
applications of the resulting resource-conscious logics of agency to
reason about the resource-sensitive manipulations of technical
artefacts.\par
\dnote{The perspectives for future research are extensions of our treatment to further non-normal modalities for logics of BDI agents. For instance, introducing resource-sensitive operators of beliefs will enable us to model agents' beliefs that depend on the amount of available information. Moreover, we are particularly interested in a resource-sensitive view of the strategic power of agents and coalitions, where the social interactions are mediated by the available resources.}


\bibliography{PoTr-revision.bbl}

\end{document}